\documentclass[a4paper,UKenglish,cleveref,thm-restate]{lipics-v2021} %

\pdfoutput=1
\usepackage{etoolbox}
\providebool{extendedversion}

\extendedversiontrue

\providebool{extendedversion}

\NewDocumentCommand\PhraseAppendix{}{%
  \ifextendedversion%
  technical appendix%
  \else%
  extended version%
  \fi%
}

\usepackage{etoolbox}
\usepackage{titletoc}
\usepackage{appendix}
\usepackage{fix-restatable}
\usepackage{jon-lipicsthm}
\usepackage{url}
\usepackage{macros}

\let\citep\cite
\NewDocumentCommand\citet{om}{\IfValueTF{#1}{#1~\citep{#2}}{\textcolor{red}{[CITET!]}}}

\Copyright{Jonathan Sterling and Robert Harper}
\EventEditors{Amy P. Felty}
\EventNoEds{1}
\EventLongTitle{7th International Conference on Formal Structures for Computation and Deduction (FSCD 2022)}
\EventShortTitle{FSCD 2022}
\EventAcronym{FSCD}
\EventYear{2022}
\EventDate{August 2--5, 2022}
\EventLocation{Haifa, Israel}
\EventLogo{}
\SeriesVolume{228}
\ArticleNo{15}

\title{Sheaf semantics of termination-insensitive noninterference}
\ifextendedversion
\subtitle{(Extended Version)}
\fi

\titlerunning{Sheaf semantics of termination-insensitive noninterference}

\author{Jonathan Sterling}{%
  Department of Computer Science, Aarhus University, Denmark%
  \and \url{https://www.jonmsterling.com}%
}{jsterling@cs.au.dk}{https://orcid.org/0000-0002-0585-5564}{}

\author{Robert Harper}{%
  Computer Science Department, Carnegie Mellon University, United States of America%
  \and \url{https://cs.cmu.edu/~rwh}%
}{rwh@cs.cmu.edu}{https://orcid.org/0000-0002-9400-2941}{}

\authorrunning{J.\ Sterling and R.\ Harper}


\ifextendedversion
\else
\relatedversion{The extended version of this paper with technical appendices for full proofs is available at \url{https://www.jonmsterling.com/papers/sterling-harper:2022.pdf}.}
\fi


\funding{
  \emph{Jonathan Sterling:} This work was supported by a Villum Investigator grant (no. 25804), Center for Basic Research in Program Verification (CPV), from the VILLUM Foundation.
  \emph{Robert Harper:} This research was sponsored by the United States Air Force Office of Scientific Research awards FA95502110009 and FA9550-21-1-0385 (Tristan Nguyen, program manager). The views and conclusions contained in this document are those of the author and should not be interpreted as representing the official policies, either expressed or implied, of any sponsoring institution, the U.S. government or any other entity.
}

\acknowledgements{We are grateful for insightful conversations with Aslan
Askarov, Stephanie Balzer, Lars Birkedal, Mart\'in Escard\'o, Marcelo Fiore,
Daniel Gratzer, and Tom de Jong. Thanks to Jamie Vicary for funding a visit to
Cambridge during which the first author learned some of the tools needed to
complete this work satisfactorily. We thank Carlos Tom\'e Corti\~{n}as, Fabian Ruch, and Sandro Stucki for
proof-reading.}

\nolinenumbers %

\begin{document}

\maketitle

\begin{abstract}
We propose a new \emph{sheaf semantics} for secure information flow over a
space of abstract behaviors, based on synthetic domain theory: security
classes are open/closed partitions, types are sheaves, and redaction of
sensitive information corresponds to restricting a sheaf to a closed subspace.
Our security-aware computational model satisfies termination-insensitive
noninterference automatically, and therefore constitutes an intrinsic
alternative to state of the art extrinsic/relational models of noninterference.
Our semantics is the latest application of Sterling and Harper's recent
re-interpretation of \emph{phase distinctions} and noninterference in
programming languages in terms of Artin gluing and topos-theoretic open/closed
modalities. Prior applications include parametricity for ML modules, the proof
of normalization for cubical type theory by Sterling and Angiuli, and the
cost-aware logical framework of Niu \etal. In this paper we employ the phase
distinction perspective \emph{twice}: first to reconstruct the syntax and
semantics of secure information flow as a lattice of phase distinctions between
``higher'' and ``lower'' security, and second to verify the computational
adequacy of our sheaf semantics with respect to a version of Abadi \etal's
\emph{dependency core calculus} to which we have added a construct for declassifying termination
channels.
 \end{abstract}

\section{Introduction}

Security-typed languages restrict the ways that classified information can flow
from high-security to low-security clients. \citet[Abadi \etal]{abadi:1999}
pioneered the use of \emph{idempotent monads} to deliver this restriction in
their \DefEmph{dependency core calculus} (DCC), parameterized in a poset of
security levels $\LVL$. Covariantly in security levels $l\in\LVL$, a family of
type operations $\alert{\TpSeal{l}{A}}$ satisfying the rules of an idempotent
monad are added to the language; the idea is then that sensitive data can be
hidden underneath $\TpSeal{l}$ and unlocked only by a client with a type that
can be equipped with a $\TpSeal{l}$-algebra structure, \ie a
\DefEmph{$\LvlPol{l}$-sealed type} in our terminology.\footnote{We use the
term ``sealing'' for what \citet[Abadi \etal]{abadi:1999} call ``protection'';
to avoid confusion, we impose a uniform terminology to encompass both our work
and that of \opcit. A final notational deviation on our part is that we will
distinguish a security level $l\in\LVL$ from the corresponding syntactical entity
$\LvlPol{l}$.} For instance, a high-security
client can read a medium-security bit:
\[
  \begin{mathblock}
    \Con{f} : \TpSeal{\Con{M}}{\TpBool}\to \TpSeal{\Con{H}}{\TpBool}\\
    \Con{f}\, u =
    x \leftarrow u; \TmSeal{\Con{H}}{\prn{\Con{not}\,x}}
  \end{mathblock}
\]

There is however no corresponding program of type $\TpSeal{\Con{H}}{\TpBool}\to
\TpSeal{\Con{M}}{\TpBool}$, because the type
$\TpSeal{\Con{M}}{\TpBool}$ of medium-security booleans is not
$\LvlPol{\Con{H}}$-sealed, \ie it cannot be equipped with the structure of a
$\TpSeal{\Con{H}}$-algebra. In fact, up to observational equivalence it is
possible to state a \DefEmph{noninterference result} that fully characterizes such programs:

\begin{proposition*}[Noninterference]
  For any closed function
  $\cdot\vdash\Con{f}:\TpSeal{\Con{H}}{\TpBool}\to\TpSeal{\Con{M}}{\TpBool}$,
  there exists a closed $\cdot\vdash b:\TpSeal{\Con{M}}\TpBool$ such that $\Con{f}\simeq \lambda\_.{b}$.
\end{proposition*}

Intuitively the noninterference result above follows because you cannot
``escape'' the monad, but to prove such a result rigorously a model
construction is needed. Today the state of the art is to employ a
\DefEmph{relational model} in the sense of Reynolds in which a type is
interpreted as a binary relation on some domain, and a term is interpreted by
a relation-preserving function. Our contribution is to
introduce an \DefEmph{intrinsic} and \DefEmph{non-relational semantics} of noninterference
presenting several advantages that we will argue for, inspired by the recent
modal reconstruction of \emph{phase distinctions} by \citet[Sterling and
Harper]{sterling-harper:2021}.

\subsection{Termination-insensitivity and the meaning of ``observation''}\label{sec:intro:tini}

\begin{notation}
  We will write $\Lift{A}$ for the lifting monad that Abadi \etal notate $A_\bot$.
\end{notation}

When we speak of noninterference up to observational equivalence, much weight
is carried by the choice of what, in fact, counts as an observation. In a
functional language with general recursion, it is conventional to say that an observation
is given by a computation of $\TpUnit$ type --- which necessarily either
diverges or converges with the unique return value $\prn{}$. Under this notion
of observation, noninterference up to observations takes a very strong
character:

\begin{quote}
  \emph{Termination-sensitive noninterference.} For a closed partial function $\cdot\vdash\Con{f}:\TpSeal{\Con{H}}\TpBool \to \Lift\,\prn{\TpSeal{\Con{M}}\TpBool}$, either $\Con{f} \simeq \lambda\_.\bot$ or there exists $\cdot\vdash b: \TpSeal{\Con{M}}\TpBool$ such that $\Con{f}\simeq\lambda\_.b$.
\end{quote}

If on the other hand we restrict observations to only terminating computations
of type $\TpBool$, we evince a more relaxed \DefEmph{termination-insensitive}
version of noninterference that allows leakage through the termination channel
but \emph{not} through the ``return channel'':
\begin{quote}
  \emph{Termination-insensitive noninterference.} For a closed partial function $\cdot\vdash\Con{f}:\TpSeal{\Con{H}}\TpBool\to\Lift\,\prn{\TpSeal{\Con{M}}\TpBool}$, given any closed $u,v$  on which $f$ terminates, we have $fu \simeq fv$.
\end{quote}

\subsection{Relational \emph{vs.}\ intrinsic semantics}\label{sec:relational-vs-intrinsic}

To verify the noninterference property for the dependency core calculus,
\citet[Abadi \etal]{abadi:1999} define a \DefEmph{relational semantics} that
starts from an insecure model of computation (domain theory \emph{qua} dcpos)
and restricts it by means of binary relations indexed in security levels that
express the indistinguishability of sensitive bits to low-security clients. The
indistinguishability relations are required to be preserved by all functions,
ensuring the security properties of the model. The relational approach has an
extrinsic flavor, being characterized by the \emph{post hoc} imposition of
order (noninterference) on an inherently disordered computational model.  We
contrast the extrinsic relational semantics of \opcit with an \emph{intrinsic}
denotational semantics in which the underlying computational model has security
concerns ``built-in'' from the start.

\subsection{Our contribution: intrinsic semantics of noninterference}

The main contribution of our paper is to develop an \emph{intrinsic semantics}
in the sense of \cref{sec:relational-vs-intrinsic}, in which
termination-insensitive noninterference (\cref{sec:intro:tini}) is not bolted
on but rather arises directly from the underlying computational model.  To
summarize our approach, instead of controlling the security properties of
ordinary dcpos using a $\LVL$-indexed logical relation, we take semantics in a
category of $\LVL$-indexed dcpos, \ie sheaves of dcpos on a space $\LvlTop$ in
which each security level $l\in \LVL$ corresponds to an open/closed partition.
Employing the viewpoint of \citet[Sterling and Harper]{sterling-harper:2021},
each of these partitions induces a \DefEmph{phase distinction} between data
visible below security level $l$ (open) and data that is hidden (closed), leading to a
novel account of the sealing monad $\TpSeal{l}$ as restriction to a closed
subspace.

Our intrinsic semantics has several advantages over the relational approach.
Firstly, termination-insensitive noninterference arises directly from our
computational model. Secondly, our model of secure information flow contributes
to the consolidation and unification of ideas in programming languages by
treating general recursion and security typing as instances of two orthogonal
and well-established notions, namely \DefEmph{axiomatic \& synthetic domain
theory} and \DefEmph{phase distinctions}/\DefEmph{Artin gluing} respectively.
Termination-insensitivity then arises from the non-trivial interaction between
these orthogonal layers.

In particular, our computational model is an instance of axiomatic domain theory
in the sense of \citet[Fiore]{fiore:1994}, and embeds into a sheaf model of
{synthetic domain
theory}~\citep{fiore-rosolini:1997,fiore-plotkin-power:1997,
fiore-plotkin:1996,fiore-rosolini:1997:cpos,fiore:1997,fiore-rosolini:2001,matache-moss-staton:2021}.
Hence the interpretation of the PCF fragment of DCC is interpreted exactly as
in the standard Plotkin semantics of general recursion in categories of partial
maps, in contrast to the relational model of Abadi~\etal.
Lastly, the view of security levels as phase distinctions per
\citet[Sterling and Harper]{sterling-harper:2021} advances a uniform
perspective on noninterference scenarios that has already proved fruitful
for resolving several problems in programming languages:
\begin{enumerate}

  \item A generalized abstraction theorem for ML modules with strong
    sums~\citep{sterling-harper:2021}.

  \item Normalization and decidability of type
    checking for cubical type theory~\citep{sterling-angiuli:2021,sterling:2021:thesis} and multi-modal type theory~\citep{gratzer:normalization:2022}; guarded canonicity for guarded dependent type theory~\citep{gratzer-birkedal:2022}.

  \item The design and metatheory of the \textbf{calf} logical
    framework~\citep{niu-sterling-grodin-harper:2022} for simultaneously
    verifying the correctness and complexity of functional programs.

\end{enumerate}

The final benefit of the phase distinction perspective is that logical
relations arguments can be re-cast as imposing an \emph{additional} orthogonal
phase distinction between \emph{syntax} and \emph{logic/specification}, an
insight originally due to Peter Freyd in his analysis of the existence and
disjunction properties in terms of Artin gluing~\citep{freyd:1978}. We employ
this insight in the present paper to develop a uniform treatment of our
denotational semantics and its computational adequacy in terms of phase
distinctions.

\section{Background: relational semantics of noninterference}\label{sec:relational-semantics}

To establish noninterference for the dependency core calculus, \citet[Abadi \etal]{abadi:1999} define a
relational model of their monadic language in which each type $A$ is interpreted as
a dcpo $\vrt{A}$ equipped with a family of admissible binary relations $R^A_l$
indexed in security levels $l\in \LVL$.  In the relational semantics, a term
$\Gamma\vdash M : A$ is interpreted as a continuous function $\Mor[\vrt{M}]{\vrt{\Gamma}}{\vrt{A}}$ such that for all $l\in \LVL$, if $\gamma \mathrel{R\Sup{\Gamma}\Sub{l}} \gamma'$ then $\vrt{M}\gamma \mathrel{R\Sup{A}\Sub{l}} \vrt{M}\gamma'$.

\begin{remark}
  Two elements $u,v\in A$ such that $u\mathrel{R}_l^Av$ have been called
  \emph{equivalent} in subsequent literature, but this terminology may lead to
  confusion as there is nothing forcing the relation to be transitive, nor even
  symmetric nor reflexive.
\end{remark}

The essence of the relational model is to impose \emph{relations} between
elements that should not be distinguishable by a certain security class; a type
like $\TpBool$ or $\Con{string}$ whose relation is totally discrete, then,
allows any security class to distinguish all distinct elements. Non-discrete
types enter the picture through the sealing modality $\TpSeal{l}$:
\[
  \vrt{\TpSeal{l}{A}} = \vrt{A}
  \qquad
  u\mathrel{R\Sup{\TpSeal{l}{A}}\Sub{k}}v \Longleftrightarrow
  \begin{cases}
    u\mathrel{R\Sup{A}\Sub{k}}v & \text{if } l \sqsubseteq k\\
    \top & \text{otherwise}
  \end{cases}
\]

Under this interpretation, the denotation of a function
$\TpSeal{\Con{H}}\TpBool\to\TpSeal{\Con{M}}{\TpBool}$ must be a
constant function, as $u\mathrel{R\Sup{\TpBool}\Sub{\Con{H}}}v$ if and only iff
$u=v$. By proving computational adequacy for this denotational semantics,
one obtains the analogous \emph{syntactic} noninterference result up to
observational equivalence.

\emph{Generalization and representation of relational semantics.} The relations
imposed on each type give rise to a form of cohesion in the sense of
\citet[Lawvere]{lawvere:2007}, where elements that are related are thought of
as ``stuck together''. Then noninterference arises from the behavior of maps
from a relatively codiscrete space into a relatively discrete space, as pointed
out by \citet[Kavvos]{kavvos:2019} in his \emph{tour de force} generalization
of the relational account of noninterference in terms of axiomatic cohesion.
Another way to understand the relational account is by \emph{representation},
as attempted by \citet[Tse and Zdancewic]{tse-zdancewic:2004} and executed by
\citet[Bowman and Ahmed]{bowman-ahmed:2015}: one may embed DCC into a
polymorphic lambda calculus in which the security abstraction is implemented by
\emph{actual} type abstraction.

\paragraph*{Adapting the relational semantics for termination-insensitivity}

In the relational semantics of the dependency core calculus, the
termination-sensitive version of noninterference is achieved by interpreting
the \emph{lift} of a type in the following way:
\[
  \vrt{A_\bot} = \vrt{A}_\bot
  \qquad
  u \mathrel{R\Sup{A_\bot}\Sub{l}} v \Longleftrightarrow
  \prn{\IsDefd{u,v} \land u \mathrel{R\Sup{A}\Sub{l}} v}
  \lor
  \prn{u=v=\bot}
\]

To adapt the relational semantics for termination-insensitivity, Abadi \etal
change the interpretation of lifts to identify \emph{all} elements with the
bottom element:
\[
  \vrt{A_\bot} = \vrt{A}_\bot
  \qquad
  u \mathrel{R\Sup{A_\bot}\Sub{l}} v \Longleftrightarrow
  \prn{\IsDefd{u,v} \land u \mathrel{R\Sup{A}\Sub{l}} v}
  \lor
  \prn{u = \bot}
  \lor
  \prn{v = \bot}
\]

That all data is ``indistinguishable'' from the non-terminating computation
means that the indistinguishability relation cannot be both transitive and
non-trivial, a somewhat surprising state of affairs that leads to our critique
of relational semantics for information flow below and motivates our new
perspective based on the analogy between \emph{phase distinctions} in
programming languages and \emph{open/closed partitions} in topological spaces~\citep{sterling-harper:2021}.

\paragraph*{Critique of relational semantics for information flow}

From our perspective there are several problems with the relational semantics
of \citet[Abadi \etal]{abadi:1999} that, while not fatal on their own, inspire
us to search for an alternative perspective.

\emph{Failure of monotonicity.}
First of all, within the context of the relational semantics it would be
appropriate to say that an object $\prn{\vrt{A},R^A\Sub{\bullet}}$ is
$\LvlPol{l}$-sealed when $R^A_l$ is the total relation. But in the semantics
of Abadi \etal, it is not necessarily the case that a $\LvlPol{l}$-sealed
object is $\LvlPol{k}$-sealed when $k\sqsubseteq l$.  It is true that objects
that are \emph{definable} in the dependency core calculus are better behaved,
but in proper denotational semantics one is not concerned with the image of an
interpretation function but rather with the entire category.

\emph{Failure of transitivity.} A more significant and harder to resolve
problem is the fact that the indistinguishability relation $R^A_l$ assigned to
each type cannot be construed as an equivalence relation --- despite the fact
that in real life, indistinguishability is indeed reflexive, symmetric, and
transitive. As we have pointed out, the adaptation of DCC's relational
semantics for termination-insensitivity is evidently incompatible with using
(total or partial) equivalence relations to model indistinguishability, as
transitivity would ensure that no two elements of $A_\bot$ can be distinguished
from another.

\emph{Where is the dominance?} Conventionally the denotational semantics for a
language with general recursion begins by choosing a category of ``predomains''
and then identifying a notion of \emph{partial map} between them that evinces a
\emph{dominance}~\citep{fiore:1994,rosolini:1986}. It
is unclear in what sense the DCC's relational semantics reflects this hard-won
arrangement; as we have seen, the adaptation of the
relational semantics for termination-insensitivity further increases the
distance from ordinary domain-theoretic semantics.

\textbf{Perspective.} Abadi \etal's relational semantics is based on imposing
secure information flow properties on an existing insecure model of partial
computation, but this is quite distinct from an \emph{intrinsic denotational
semantics} for secure information flow --- which would necessarily entail new
notions of predomain and partial map that are sensitive to security from the
start. In this paper we report on such an intrinsic semantics for secure
information flow in which termination-insensitive noninterference arises
inexorably from the chosen dominance.

\section{Central ideas of this paper}

In this section, we dive a little deeper into several of the main concepts that
substantiate the contributions of this paper.  We begin by fixing a poset
$\LVL$ of security levels closed under finite meets, for example $\LVL =
\brc{\LvlLow\sqsubset\LvlMed\sqsubset\LvlHigh\sqsubset\LvlGod}$. The
purpose of including a security level even higher than $\LvlHigh$ will become
apparent when we explain the meaning of the sealing monad $\TpSeal{l}$.

\begin{notation}
  Given a space $\XTop$ and an open set $U\in\Opns{\XTop}$, we will write
  $\Sl{\XTop}{U}$ for the open subspace spanned by $U$ and
  $\ClSubcat{\XTop}{U}$ for the corresponding complementary closed subspace.
  We also will write $\Sh{\XTop}$ for the category of sheaves on the space $\XTop$.
\end{notation}

\subsection{A space of abstract behaviors and security policies}
We begin by transforming the security poset $\LVL$ into a topological space
$\LvlTop$ of ``abstract behaviors'' whose algebra of open sets
$\Opns{\LvlTop}$ can be thought of as a lattice of \emph{security policies}
that govern whether a given behavior is permitted.

\begin{definition}
  An \DefEmph{abstract behavior} is a \emph{filter} on the poset $\LVL$, \ie
  a monotone subset $x\subseteq\LVL$ such that $\Conj{i<n}{l_i}\in x$ if and only if each
  $l_i\in x$.
\end{definition}

\begin{definition}
  A \DefEmph{security policy} is a lower set in $\LVL$, \ie an antitone subset
  $U\subseteq\LVL$. We will write $\alert{U\Vdash x}$ to mean $U$ permits the
  behavior $x$, \ie the subset $x\cap U$ is inhabited.
\end{definition}

An abstract behavior $x$ denotes the set of security levels $l\in\LVL$ at which
it is permitted; a security policy $U$ denotes the set of
security levels \emph{above which} some behavior is permitted.

\begin{construction}
  We define $\LvlTop$ to be the topological space whose points are abstract
  behaviors, and whose open sets are of the form $\Compr{x}{U\Vdash x}$ for
  some security policy $U$.\footnote{Those familiar with the point-free
  topology of topoi~\citep{johnstone:2002,vickers:2007,anel-joyal:2021} will
  recognize that $\LvlTop$ is more simply described as the presheaf topos
  $\PrTop{\LVL}$: viewed as a space, it is the dcpo completion of
  $\OpCat{\LVL}$, and as a frame it is the free cocompletion of $\LVL$.  The
  definition of $U\Vdash x$ then presents a computation of the \emph{stalk}
  $U_x$ of the subterminal sheaf $U\in\Sh{\LvlTop}$ at the behavior
  $x\in\LvlTop$.}
\end{construction}

We have a meet-preserving embedding of posets
$\EmbMor[\LvlPol{-}]{\LVL}{\Opns{\LvlTop}}$ that exhibits $\Opns{\LvlTop}$
as the free completion of $\LVL$ under joins.

\begin{intuition}[Open and closed subspaces]
  Each security level $l\in\LVL$ represents a security policy
  $\LvlPol{l}\in\Opns{\LvlTop}$ whose corresponding open subspace
  $\Sl{\LvlTop}{\LvlPol{l}}$ is spanned by the behaviors \emph{permitted} at
  security levels $l$ and above.  Conversely the complementary closed subspace
  $\ClSubcat{\LvlTop}{\LvlPol{l}} = \LvlTop\setminus\Sl{\LvlTop}{\LvlPol{l}}$
  is spanned by behaviors that are \emph{forbidden} at security level $l$ and
  below.
\end{intuition}

\subsection{Sheaves on the space of abstract behaviors}

Our intention is to interpret each type of a dependency core calculus as a
\emph{sheaf} on the space $\LvlTop$ of abstract behaviors. To see why this
interpretation is plausible as a basis for secure information flow, we note
that a sheaf on $\LvlTop$ is the same thing as a presheaf on the poset $\LVL$,
\ie a family of sets $\prn{A\Sub{l}}\Sub{l\in\LVL}$ indexed contravariantly in
$\LVL$ in the sense that for $k\sqsubseteq l$ there is a chosen restriction
function $A_l\to A_k$ satisfying two laws. Hence a sheaf on $\LvlTop$ determines
(1) for each security level $l\in\LVL$ a choice of what data is visible under the
security policy $\LvlPol{l}$, and (2) a way to \emph{redact} data as it passes under
a more restrictive security policy $\LvlPol{k}\subseteq\LvlPol{l}$.

\subsection{Transparency and sealing from open and closed subspaces}\label{sec:key-ideas:modalities}

For any subspace $\TopIdent{Q}\subseteq\LvlTop$, a sheaf $A\in\Sh{\LvlTop}$ can
be restricted to $\TopIdent{Q}$, and then extended again to $\LvlTop$. This
composite operation gives rise to an \emph{idempotent monad} on $\Sh{\LvlTop}$
that has the effect of purging any data from $A\in\Sh{\LvlTop}$ that cannot be
seen from the perspective of $\TopIdent{Q}$.  The idempotent monads
corresponding to the open and closed subspaces induced by a security level
$l\in\LVL$ are named and notated as follows:
\begin{enumerate}
  \item The \DefEmph{transparency monad} $A\mapsto
    \alert{\prn{\OpMod{\LvlPol{l}}{A}}}$ replaces $A$ with whatever part of
    it can be viewed under the policy $\LvlPol{l}$. The transparency monad is the function space $A\Sup{\LvlPol{l}}$, recalling that an open set of $\LvlTop$ is the same as a subterminal sheaf. When the unit
    is an isomorphism at $A$, we say that $A$ is
    \DefEmph{$\LvlPol{l}$-transparent}.
  \item The \DefEmph{sealing monad} $A\mapsto
    \alert{\prn{\ClMod{\LvlPol{l}}{A}}}$ removes from $A$ whatever part of it
    can be viewed under the policy $\LvlPol{l}$. The sealing monad can be constructed as the pushout $\LvlPol{l}\sqcup\Sub{\LvlPol{l}\times A}A$. When the unit
    is an isomorphism at $A$, we say that $A$ is
    \DefEmph{$\LvlPol{l}$-sealed}.
\end{enumerate}

\begin{figure}
  \begingroup
  \tikzset{every node/.append style={opacity = 1}}
  \[
    \begin{tikzpicture}[scale=.5,block/.append style = {fill=RegalBlue!60},baseline=(current bounding box.center)]
      \draw[block,fill opacity=.2] (0,0) rectangle node {$A\Sub{\LvlLow}$} (2,1);
      \draw[block,fill opacity=.4] (0,1) rectangle node {$A\Sub{\LvlMed}$} (2,2);
      \draw[block,fill opacity=.6] (0,2) rectangle node {$A\Sub{\LvlHigh}$} (2,3);
      \draw[block,fill opacity=.8] (0,3) rectangle node {$A\Sub{\LvlGod}$} (2,4);
      \draw[draw=none] (0,-1) rectangle node {$\Sh{\LvlTop}\ni A$} (2,0);
    \end{tikzpicture}
    \mapsto
    \begin{tikzpicture}[scale=.5,block/.append style = {fill=RegalBlue!60,fill opacity = .2},baseline=(current bounding box.center)]
      \draw[draw=none] (0,0) rectangle (2,4);
      \draw[draw=none] (0,0) rectangle (2,3);

      \draw[block,fill opacity = 0.2] (0,0) rectangle node {$A\Sub{\LvlLow}$} (2,1);
      \draw[block,fill opacity = 0.4] (0,1) rectangle node {$A\Sub{\LvlMed}$} (2,2);
      \draw[draw=none] (0,-1) rectangle node {$\vphantom{\Sh{\LvlTop}}\smash{\Sh{\Sl{\LvlTop}{\LvlPol{\LvlMed}}}}$} (2,0);
    \end{tikzpicture}
    \mapsto
    \begin{tikzpicture}[scale=.5,block/.append style = {fill=RegalBlue!60,fill opacity = .2},baseline=(current bounding box.center)]
      \draw[block,fill opacity = 0.2] (0,0) rectangle node {$A\Sub{\LvlLow}$} (2,1);
      \draw[block,fill opacity = 0.4] (0,1) rectangle node {$A\Sub{\LvlMed}$} (2,2);
      \draw[block,fill opacity = 0.4] (0,2) rectangle node {$A\Sub{\LvlMed}$} (2,3);
      \draw[block,fill opacity = 0.4] (0,3) rectangle node {$A\Sub{\LvlMed}$} (2,4);
      \draw[draw=none] (0,-1) rectangle node {$\Sh{\LvlTop}\ni \alert{\OpMod{\LvlPol{\LvlMed}}{A}}$} (2,0);
    \end{tikzpicture}
    \qquad
    \begin{tikzpicture}[scale=.5,block/.append style = {fill=RegalBlue!60,fill opacity = .2},baseline=(current bounding box.center)]
      \draw[block,fill opacity=0.2] (0,0) rectangle node {$A\Sub{\LvlLow}$} (2,1);
      \draw[block,fill opacity=0.4] (0,1) rectangle node {$A\Sub{\LvlMed}$} (2,2);
      \draw[block,fill opacity=0.6] (0,2) rectangle node {$A\Sub{\LvlHigh}$} (2,3);
      \draw[block,fill opacity=0.8] (0,3) rectangle node {$A\Sub{\LvlGod}$} (2,4);
      \draw[draw=none] (0,-1) rectangle node {$\Sh{\LvlTop}\ni A$} (2,0);
    \end{tikzpicture}
    \mapsto
    \begin{tikzpicture}[scale=.5,block/.append style = {fill=RegalBlue!60,fill opacity = .2},baseline=(current bounding box.center)]
      \draw[draw=none] (0,0) rectangle (2,4);
      \draw[draw=none] (0,0) rectangle (2,3);
      \draw[block,fill opacity=.8] (0,3) rectangle node {$A\Sub{\LvlGod}$} (2,4);
      \draw[block,fill opacity=.6] (0,2) rectangle node {$A\Sub{\LvlHigh}$} (2,3);
      \draw[draw=none] (0,-1) rectangle node {$\vphantom{\Sh{\LvlTop}}\smash{\Sh{\ClSubcat{\LvlTop}{\LvlPol{\LvlMed}}}}$} (2,0);
    \end{tikzpicture}
    \mapsto
    \begin{tikzpicture}[scale=.5,block/.append style = {fill=RegalBlue!60,fill opacity = .2},baseline=(current bounding box.center)]
      \draw[block,fill=gray,fill opacity=.05] (0,0) rectangle node {$\brc{\SlPt}$} (2,1);
      \draw[block,fill=gray,fill opacity=.05] (0,1) rectangle node {$\brc{\SlPt}$} (2,2);
      \draw[block,fill opacity=.6] (0,2) rectangle node {$A\Sub{\LvlHigh}$} (2,3);
      \draw[block,fill opacity=.8] (0,3) rectangle node {$A\Sub{\LvlGod}$} (2,4);
      \draw[draw=none] (0,-1) rectangle node {$\Sh{\LvlTop}\ni\alert{\ClMod{\LvlPol{\LvlMed}}{A}}$} (2,0);
    \end{tikzpicture}
  \]
  \endgroup
  \caption{The transparency and sealing monads for $\LvlMed\in\LVL$  on a sheaf $A\in\Sh{\LvlTop}$ visualized.}
  \label{fig:transp-and-sealing-viz}
\end{figure}

The transparency and sealing monads interact in two special ways,
which can be made apparent by appealing to the visualization of their behavior
that we present in \cref{fig:transp-and-sealing-viz}.
\begin{enumerate}
  \item The $\LvlPol{l}$-transparent part of a $\LvlPol{l}$-sealed sheaf is trivial, \ie we have $\prn{\OpMod{\LvlPol{l}}{\prn{\ClMod{\LvlPol{l}}{A}}}}\cong \brc{\SlPt}$.

  \item Any sheaf $A\in\Sh{\LvlTop}$ can be reconstructed as the fiber product $\prn{\OpMod{\LvlPol{l}}{A}}\times\Sub{\ClMod{\LvlPol{l}}{\prn{\OpMod{\LvlPol{l}}{A}}}}\ClMod{\LvlPol{l}}{A}$.
\end{enumerate}

The first property above immediately gives rise to a form of noninterference, which justifies our intent to interpret DCC's sealing monad as $\alert{\TpSeal{l}{A} = \ClMod{\LvlPol{l}}{A}}$.
\begin{observation}[Noninterference]\label{obs:naive-noninterference}
  Any map $\Mor{\ClMod{\LvlPol{l}}{A}}{\Con{bool}}$ is constant.
\end{observation}
\begin{proof}
  We may verify that the boolean sheaf $\TpBool$ is $\LvlPol{l}$-transparent
  for all $l\in\LVL$.
\end{proof}

Our sealing monad above is well-known to the
type-and-topos--theoretic community as the \DefEmph{closed
modality}~\citep{rijke-shulman-spitters:2020,schultz-spivak:2019,sga:4}
corresponding to the open set $\LvlPol{l}\in\Opns{\LvlTop}$. In the context
of (total) dependent type theory, our sealing monad has excellent properties
not shared by those of \citet[Abadi \etal]{abadi:1999}, such as justifying
\emph{dependent} elimination rules and commuting with identity types.  In
contrast to the \emph{classified sets} of \citet[Kavvos]{kavvos:2019} which
cannot form a topos, our account of information flow is compatible with the
full internal language of a topos.

\subsection{Recursion and termination-insensitivity via sheaves of domains}\label{sec:key-ideas:recursion}

To incorporate recursion into our sheaf semantics of information flow, in this
section we consider \emph{internal dcpos} in $\Sh{\LvlTop}$, \ie sheaves of
dcpos. Later in the technical development of our paper, we work in the
axiomatic setting of synthetic domain theory, but all the necessary intuitions
can also be understood concretely in terms of dcpos.
Domain theory internal to $\Sh{\LvlTop}$ works very similarly to classical
domain theory, but it must be developed without appealing to the law of the
excluded middle or the axiom of choice as these do not hold in $\Sh{\LvlTop}$
except for a particularly degenerate security poset. \citet[De Jong and
Escard\'o]{dejong-escardo:2021} explain how to set up the basics of domain
theory in a suitably constructive manner, which we will not review.

The sheaf-theoretic domain semantics sketched above leads immediately to a new and
simplified account of termination-insensitivity. It is instructive to consider
whether there is an analogue to \cref{obs:naive-noninterference} for partial
continuous functions $\alert{\Mor{\ClMod{\LvlPol{l}}{A}}{\Lift\,{\TpBool}}}$.
It is not the case that $\Lift\,\TpBool$ is $\LvlPol{l}$-sealed for all
$l\in\LVL$, so it would not follow that any continuous map
$\Mor{\ClMod{\LvlPol{l}}{A}}{\Lift\,{\TpBool}}$ is constant. A partial function always extends to a total function on a restricted domain, however, so we may immediately conclude the following:
\begin{observation}[Termination-insensitive noninterference]
  For any continuous map $\Mor[f]{\ClMod{\LvlPol{l}}{A}}{\Lift\,{\TpBool}}$
  and elements $u,v:\ClMod{\LvlPol{l}}{A}$ with $fu$ and $fv$
  defined, we have $fu = fv$.
\end{observation}

This is the sense in which termination-insensitive noninterference arises
automatically from the combination of domain theory with sheaf semantics for
information flow.

\section{Refined dependency core calculus}\label{sec:calculus}

We now embark on the technical development of this paper, beginning with
a call-by-push-value (cbpv) style~\citep{levy:2004} refinement of the dependency core calculus
over a poset $\LVL$ of security levels. We will work informally in
the logical framework of locally Cartesian closed categories \emph{\`a la}
\citet[Gratzer and Sterling]{gratzer-sterling:2020}; we will write $\TCat$ for the free locally
Cartesian closed category generated by all the constants and equations
specified herein.%

\NewDocumentCommand\RuleBlockJudgmental{}{
  \Tp\Pos,\Tp\Neg : \LfSort\mathbreak
  \Tm : \Tp\Pos\to\LfSort\mathbreak
  \TpU : \Tp\Neg\to\Tp\Pos\mathbreak
  \TpF : \Tp\Pos\to\Tp\Neg
}

\NewDocumentCommand\RuleBlockRetBind{}{
  \TmRet : A \to \TpU\TpF{A}\mathbreak
  \TmBind : \TpU\TpF{A} \to \prn{A\to \TpU{X}} \to \TpU{X}
}

\NewDocumentCommand\RuleBlockBindEqns{}{
  \TmBind\,\prn{\TmRet\,{u}}\,f \equiv\Sub{\TpU{X}} f\,u\mathbreak
  \TmBind\, u\, \TmRet \equiv\Sub{\TpU\TpF{A}} u\mathbreak
  \TmBind\, \prn{\TmBind\, u\, f}\, g \equiv\Sub{\TpU{X}} \TmBind\, u\, \prn{\lambda x. \TmBind\, \prn{f\,x}\, g}
}

\NewDocumentCommand\RuleBlockFix{}{
  \TmFix : \prn{\TpU{X}\to \TpU{X}}\to \TpU{X}\mathbreak
  \TmFix\, {f} \equiv f\,\prn{\TmFix\, f}
}

\NewDocumentCommand\RuleBlockFn{}{
  \TpFn : \Tp\Pos\to\Tp\Neg\to\Tp\Neg\mathbreak
  \TpFn.\Tm : \prn{A \to \TpU{X}}\cong \TpU\,\prn{\TpFn\,A\,X}
}

\NewDocumentCommand\RuleBlockProd{}{
  \TpProd : \Tp\Pos\to\Tp\Pos\to\Tp\Pos\mathbreak
  \TpProd.\Tm : {A \times B}\cong {\TpProd\,A\,B}
}

\NewDocumentCommand\RuleBlockUnit{}{
  \TpUnit : \Tp\Pos\mathbreak
  \TpUnit.\Tm : \ObjTerm\cong \TpUnit
}

\NewDocumentCommand\RuleBlockIsSealed{}{
  \SealedBelow{l} : \Tp\Pos\to\LfProp\mathbreak
  \SealedBelow{l}\,A \coloneqq \LvlPol{l} \to \Compr{x:A}{\forall y:A. x\equiv_A y}
}

\NewDocumentCommand\RuleBlockSeal{}{
  \TpSeal{l} : \Tp\Pos\to\SealedTp{l}\Pos\mathbreak
  \TmSeal{l} : A\to \TpSeal{l}\,A\mathbreak
}

\NewDocumentCommand\RuleBlockUnseal{}{
  \TmUnseal{l} : \brc{B : \SealedTp{l}\Pos} \to \TpSeal{l}\,{A} \to \prn{A \to B} \to B\mathbreak
  \TmUnseal{l}\, \prn{\TmSeal{l}\, u}\, f \equiv\Sub{B} f\,{u}\mathbreak
  \TmUnseal{l}\,u\,\prn{\lambda x. f\, \prn{\TmSeal{l}\,x}} \equiv\Sub{B} f\,{u}
}

\NewDocumentCommand\RuleBlockSum{}{
  \TpSum : \Tp\Pos\to \Tp\Pos\to\Tp\Pos\mathbreak
  \TmInl : A \to \TpSum\,A\,B\mathbreak
  \TmInr : B \to \TpSum\,A\,B
}

\NewDocumentCommand\RuleBlockCase{}{
  \TmCase : \TpSum\,A\,B\to \prn{A\to C}\to\prn{B\to C} \to C\mathbreak
  \TmCase\,\prn{\TmInl\, u}\,f\,g \equiv\Sub{C} f\,u\mathbreak
  \TmCase\,\prn{\TmInr\, v}\,f\,g \equiv\Sub{C} g\,v\mathbreak
  \TmCase\,u\,\prn{\lambda x. f\,\prn{\TmInl\,x}}\,\prn{\lambda x.f\,\prn{\TmInr\,x}} \equiv\Sub{C} f\, u
}

\NewDocumentCommand\RuleBlockTdcl{}{
  \TmTdcl{l} : \brc{A : \SealedTp{l}\Pos} \to \TpSeal{l}\TpU\TpF{A} \to \TpU\TpF{A}\mathbreak
  \TmTdcl{l}\, \prn{\TmSeal{l}\, \prn{\TmRet\, u}} \equiv\Sub{\TpU\TpF{A}} \TmRet\, u
}

\NewDocumentCommand\RuleBlockLvl{}{
  \LvlPol{l} : \LfProp\mathbreak
  \LvlPol{k}\to \LvlPol{l}&\prn{k\leq l\in\LVL}\mathbreak
  \LvlPol{k}\to \LvlPol{l}\to \LvlPol{k\land l}
}

\subsection{The basic language}

We have value types $A:\Tp\Pos$ and computation types $X:\Tp\Neg$; because our
presentation of cbpv does not include stacks, we will not include a separate
syntactic category for computations but instead access them through thunking.
The sorts of value and computation types and their adjoint connectives are
specified below:
\[
  \begin{mathblock}
    \def\mathbreak{\qquad}
    \RuleBlockJudgmental
  \end{mathblock}
\]

We let $A,B,C$ range over $\Tp\Pos$ and $X,Y,Z$ over $\Tp\Neg$. We will often
write $A$ instead of $\Tm\,{A}$ when it causes no ambiguity. Free computation types are
specified as follows:
\[
  \begin{mathblock}
    \RuleBlockRetBind
  \end{mathblock}
  \qquad
  \begin{mathblock}
    \RuleBlockBindEqns
  \end{mathblock}
\]

We support general recursion in computation types:
\[
  \begin{mathblock}
    \def\mathbreak{\qquad}
    \RuleBlockFix
  \end{mathblock}
\]

We close the universe $X:\Tp\Neg \vdash \Tm\,{\TpU{X}}$ of computation types
and thunked computations under all function types $\Tm\,{A}\to \Tm\,{\TpU{X}}$
by adding a new computation type constant $\TpFn$ equipped with a universal
property like so:
\[
  \begin{mathblock}
    \def\mathbreak{\qquad}
    \RuleBlockFn
  \end{mathblock}
\]

We will treat this isomorphism implicitly in our informal notation, writing
$\lambda x.u\prn{x}$ for both meta-level and object-level function terms.
Finite product types are specified likewise:
\[
  \begin{mathblock}
    \RuleBlockProd
  \end{mathblock}
  \qquad
  \begin{mathblock}
    \RuleBlockUnit
  \end{mathblock}
\]

Sum types must be treated specially because we do not intend them to be
coproducts in the logical framework: they should have a universal property for
types, not for sorts.
\[
  \begin{mathblock}
    \RuleBlockSum
  \end{mathblock}
  \qquad
  \begin{mathblock}
    \RuleBlockCase
  \end{mathblock}
\]

\subsection{The sealing modality and declassification}\label{sec:calculus:sealing}

For each $l\in\LVL$, we add an \emph{abstract} proof irrelevant proposition $\LvlPol{l} :
\LfProp$ to the language; this proposition represents the condition that the
``client'' has a lower security clearance than $l$. This ``redaction'' is
implemented by isolating the types that are \emph{sealed} at $\LvlPol{l}$,
\ie those that become singletons in the presence of $\LvlPol{l}$:
\[
  \begin{mathblock}
    \begin{mathblock}
      \RuleBlockLvl
    \end{mathblock}
    \qquad
    \begin{mathblock}
      \RuleBlockIsSealed
    \end{mathblock}
  \end{mathblock}
\]

We will write $\SealedTp{l}\Pos\subseteq\Tp\Pos$ for the subtype spanned by
value types $A$ for which $\SealedBelow{l}\,A$ holds.  As in
\cref{sec:key-ideas:modalities}, we will write $\SlPt$ for the unique element
of an $\LvlPol{l}$-sealed type in the presence of $u:\LvlPol{l}$.
Next we add the sealing modality itself:
\[
  \begin{mathblock}
    \RuleBlockSeal
  \end{mathblock}
  \qquad
  \begin{mathblock}
    \RuleBlockUnseal
  \end{mathblock}
\]

Finally a construct for declassifying the termination channel of a sealed computation:
\[
  \begin{mathblock}
    \def\mathbreak{\qquad}
    \RuleBlockTdcl
  \end{mathblock}
\]

\begin{remark}
  The $\LvlPol{l}$ propositions play a purely book-keeping role, facilitating verification of program equivalences in the same sense as
  the ghost variables of \citet[Owicki and Gries]{owicki-gries:1976}.
\end{remark}
\section{Denotational semantics in synthetic domain theory}\label{sec:denotational-semantics}

We will define our denotational semantics for information flow and
termination-insensitive noninterference in a category of domains indexed in
$\LVL$. To give a model of the theory presented in \cref{sec:calculus} means to
define a locally Cartesian closed functor $\Mor{\TCat}{\ECat}$ where $\ECat$ is
locally Cartesian closed. Unfortunately no category of domains can be locally
Cartesian closed, but we can \emph{embed} categories of domains in a locally
Cartesian closed category by following the methodology of \DefEmph{synthetic
domain theory}~\citep{fiore-rosolini:1997,fiore-plotkin-power:1997,
fiore-plotkin:1996,fiore-rosolini:1997:cpos,fiore:1997,fiore-rosolini:2001,matache-moss-staton:2021}.\footnote{In
particular we focus on the style of synthetic domain theory based on
Grothendieck topoi and well-complete objects. There is another very productive
strain of synthetic domain theory based on realizability and replete objects
that has different
properties~\citep{hyland:1991,phoa:1991,taylor:1991,reus:1995,reus:1996,reus:1999,reus-streicher:1993,reus-streicher:1999}.}

\subsection{A topos for information flow logic}

Recall that $\LVL$ is a poset of security levels closed under finite meets. The
presheaf topos $\LvlTop$ defined by the identification $\Sh{\LvlTop} =
\brk{\OpCat{\LVL},\SET}$ contains propositions $\Yo[\LVL]{l}$ corresponding to
every security level $l\in\LVL$, and is closed under both sealing and
transparency modalities $\OpMod{\Yo[\LVL]{l}}E,\ClMod{\Yo[\LVL]{l}}{E}$ in the
sense of \cref{sec:key-ideas:modalities}; in more traditional parlance, these are the
\emph{open} and \emph{closed} modalities corresponding to the proposition
$\Yo[\LVL]{l}$~\citep{rijke-shulman-spitters:2020}. It is possible to give a
denotational semantics for a \emph{total} fragment of our language in
$\Sh{\LvlTop}$, but to interpret recursion we need some kind of domain theory. We
therefore define a topos model of synthetic domain theory that lies over
$\LvlTop$ and hence incorporates the information flow modalities seamlessly.

\subsection{Synthetic domain theory over the information flow topos}

We will now work abstractly with a Grothendieck topos $\CmpTop$ equipped with a
dominance $\Sigma\in\Sh{\CmpTop}$, called the \DefEmph{Sierpi\'nski space},
satisfying several axioms that give rise to a reflective subcategory of objects
that behave like predomains.  We leave the construction of $\CmpTop$ to our
\PhraseAppendix, where it is built by adapting the recipe of \citet[Fiore
and Plotkin]{fiore-plotkin:1996}.

\begin{definition}[{\citet[Rosolini]{rosolini:1986}}]
  A \DefEmph{dominion} on a category $\ECat$ is a stable class of monos
  closed under identity and composition. Given a dominion $\mathcal{M}$ such
  that $\ECat$ has finite limits, a \DefEmph{dominance} for $\mathcal{M}$ is a
  classifier $\Mor|>->|[\top]{\ObjTerm{\ECat}}{\Sigma}$ for the elements of
  $\mathcal{M}$ in the sense that every
  $\Mor|>->|{U}{A}\in\mathcal{M}$ gives rise to a \emph{unique} map
  $\Mor[\chi_U]{A}{\Sigma}$ such that $U \cong \chi_U^*\top$.
\end{definition}

\NewDocumentCommand\ShLub{}{\bar{\boldsymbol{\omega}}}
\NewDocumentCommand\ShChain{}{\boldsymbol{\omega}}
\NewDocumentCommand\InclChain{}{\boldsymbol{\iota}}

If $\ECat$ is locally cartesian closed, we may form the \DefEmph{partial element
classifier} monad $\Mor[\Lift]{\ECat}{\ECat}$ for a dominance $\Sigma$, setting
$\Lift{E} = \Sum{\phi:\Sigma} \OpMod{\phi}{E}$; given $e\in\Lift{E}$, we will
write $\IsDefd{e}\in\Sigma$ for the termination support $\pi_1\,e$ of $e$. We
are particularly interested in the case where $\Lift$ has a final coalgebra
$\ShLub\cong\Lift{\ShLub}$ and an initial algebra
$\Lift{\ShChain}\cong\ShChain$.
When $\ECat$ is the category of sets, $\ShChain$ is just the natural numbers
object $\mathbb{N}$ and $\ShLub$ is $\mathbb{N}_\infty$, the natural numbers
with an infinite point adjoined. In general, one should think of $\ShChain$ as
the ``figure shape'' of a formal $\omega$-chain $\Mor{\ShChain}{E}$ that takes
into account the data of the dominance; then $\ShLub$ is the figure shape of a
formal $\omega$-chain equipped with its supremum, given by evaluation at the
infinite point $\infty\in\ShLub$. There is a canonical inclusion
$\Mor|>->|[\InclChain]{\ShChain}{\ShLub}$ witnessing the \emph{incidence
relation} between a chain equipped with its supremum and the underlying chain.

\begin{restatable}{sdtaxiom}{AxSigmaFiniteJoins}\label{axiom:sigma-finite-joins}
  $\Sigma$ has \DefEmph{finite joins} $\Disj{i<n}\phi_i$ that are preserved by
  the inclusion $\Sigma\subseteq\Omega$. We will write $\bot$ for the empty
  join and $\phi\lor\psi$ for binary joins.
\end{restatable}

\begin{definition}[Complete types]
  In the internal language of $\ECat$, a type $E$ is called \DefEmph{complete}
  when it is internally orthogonal to the comparison map $\Mor|>->|{\ShChain}{\ShLub}$. In
  the internal language, this says that for any formal chain
  $\Mor[e]{\ShChain}{E}$ there exists a \emph{unique} figure
  $\Mor[\hat{e}]{\ShLub}{E}$ such that $\hat{e}\circ\InclChain = e$. In this
  scenario, we write $\DLub{i\in\ShChain}{e\,i}$ for the evaluation
  $\hat{e}\,\infty$.
\end{definition}

\begin{restatable}{sdtaxiom}{AxOmegaInductive}\label{axiom:omega-inductive}
  The initial lift algebra $\ShChain$ is the colimit of the following
  $\omega$-chain of maps:
  \[
    \begin{tikzpicture}[diagram]
      \node (0) {$\ObjInit$};
      \node (1) [right = of 0] {$\Lift{\ObjInit}$};
      \node (2) [right = of 1] {$\Lift^2{\ObjInit}$};
      \node (3) [right = of 2] {$\ldots$};
      \draw[->] (0) to node [above] {$!$} (1);
      \draw[->] (1) to node [above] {$\Lift{!}$} (2);
      \draw[->] (2) to node [above] {$\Lift^2{!}$} (3);
    \end{tikzpicture}
  \]
\end{restatable}

\begin{definition}
  A type $E$ is called a \DefEmph{predomain} when $\Lift{E}$ is complete.
\end{definition}

\begin{restatable}{sdtaxiom}{AxSigmaPredomain}\label{axiom:sigma-predomain}
  The dominance $\Sigma$ is a predomain.
\end{restatable}

The category of predomains is complete, cocomplete,
closed under lifting, exponentials, and powerdomains, and is a reflective
exponential ideal in $\Sh{\CmpTop}$ --- thus better behaved than any classical
category of predomains. The predomains with $\Lift$-algebra structure serve as
an appropriate notion of \emph{domain} in which arbitrary fixed points can be
interpreted by taking the supremum of formal $\ShChain$-chains of
approximations $f^n\bot$; in addition to ``term-level'' recursion, we may also
interpret recursive types.
We impose two additional axioms for information flow:

\begin{restatable}{sdtaxiom}{AxLvls}\label{axiom:lvls}
  The topos $\CmpTop$ is equipped with a geometric morphism
  $\Mor[\LvlMapCmp]{\CmpTop}{\LvlTop}$ such that the induced functor
  $\Mor[\LvlMapCmp^*\Yo[\LVL]]{\LVL}{\Opns{\CmpTop}}$ is fully faithful and is
  valued in $\Sigma$-propositions. We will write $\gl{l}$ for each
  $\LvlMapCmp^*\Yo[\LVL]{l}$.
\end{restatable}

\cref{axiom:lvls} ensures that our domain theory include computations whose
termination behavior depends on the observer's security level.
The following \cref{axiom:finset-transparent-predomain} is applied to the semantic
noninterference property.

\begin{restatable}{sdtaxiom}{AxFinTransparentPredomain}\label{axiom:finset-transparent-predomain}
  Any constant object $\CmpTop^*\brk{n}\in\Sh{\CmpTop}$ for $\brk{n}$ a finite
  set is an $\LvlPol{l}$-transparent predomain for any $l\in\LVL$.
\end{restatable}

The category $\Sh{\CmpTop}$ is closed under as many topos-theoretic
universes~\citep{streicher:2005} as there are Grothendieck universes in the
ambient set theory. For any such universe $\UU_i$, there is a subuniverse
$\DDPos_i\subseteq\UU_i$ spanned by \emph{predomains}; we note that being a
predomain is a property and not a structure. The object $\DDPos_i$ can exist
because being a predomain is a local property that can be expressed in the
internal logic. In fact, the predomains can be seen to be not only a reflective
subcategory but also a reflective \emph{subfibration} as they are obtained by
the internal localization at a class of
maps~\citep{shulman:blog:reflective-subfibrations}; therefore the reflection
can be internalized as a connective $\Mor{\UU_i}{\DDPos_i}$ implemented as a
quotient-inductive type~\citep{shulman:blog:localization-hit}.
We may define the corresponding universe of domains $\DDNeg_i$ to be the
collection of predomains in $\DDPos_i$ equipped with $\Lift$-algebra
structures. We hereafter suppress universe levels.

\subsection{The stabilizer of a predomain and its action}\label{sec:stabilizer}

In this section, we work internally to the synthetic domain theory of
$\Sh{\CmpTop}$; first we recall the definition of an \emph{action} for a
commutative monoid.

\begin{definition}
  Let $\prn{M,0,+}$ be a monoid object in the category of predomains; an
  \DefEmph{$M$-action structure} on a predomain $A$ is given by a function
  $\Mor[\mathord{\parallel_A}]{M\times A}{A}$ satisfying the identities $0\parallel_A a = a$ and $m \parallel_A
  n\parallel_A a = \prn{m + n}\parallel_A a$.
\end{definition}

Write $\Sigma^\lor$ for the additive monoid structure of the Sierpi\'nski
domain, with addition given by $\Sigma$-join $\phi\lor\psi$ and the
unit given by the non-terminating computation $\bot$. Our terminology below is
inspired by stabilizer subgroups in algebra.

\begin{definition}[The stabilizer of a predomain]\label{def:stabilizer}
  Given a predomain $A$, we define the \DefEmph{stabilizer} of $A$ to be the
  submonoid $\Stab{A}\subseteq\Sigma^\lor$ spanned by $\phi : \Sigma^\lor$ such that $A$ is
  $\phi$-sealed, \ie the projection map $\Mor{A\times \phi}{\phi}$ is an
  isomorphism.
\end{definition}

\begin{remark}
  We can substantiate the analogy between \cref{def:stabilizer} and stabilizer
  subgroups in algebra. Up to coherence issues that could be solved using
  higher categories, any category $\CatIdent{P}$ of predomains closed under
  subterminals and pushouts can be structured with a monoid action over
  $\Sigma^\lor$; the action
  $\Mor[\mathord{\parallel\Sub{\CatIdent{P}}}]{\Sigma^\lor\times\CatIdent{P}}{\CatIdent{P}}$
  takes $A$ to the $\phi$-sealed object $\phi\parallel\Sub{\CatIdent{P}} A
  \coloneqq \ClMod{\phi}{A}$. Up to isomorphism, the identities for a
  $\Sigma^\lor$-action can be seen to be satisfied. Then we say that the
  stabilizer of a predomain $A\in\CatIdent{P}$ is the submonoid
  $\Stab{A}\subseteq\Sigma^\lor$ consisting of propositions $\phi$ such that
  $\phi\parallel\Sub{\CatIdent{P}} {A} \cong {A}$.
\end{remark}

\begin{lemma}\label{lem:stabilizer-action}
  For any predomain $A$, we may define a canonical $\Stab{A}$-action on
  $\Lift{A}$:
  \begin{align*}
    \parallel\Sub{\Lift{A}} &: \Stab{A}\times \Lift{A}\to\Lift{A}\\
    \phi \parallel\Sub{\Lift{A}} a &=
    \prn{
      \phi \lor \IsDefd{a},
      \brk{
        \phi\hookrightarrow \SlPt,
        \IsDefd{a}\hookrightarrow a
      }
    }
  \end{align*}
\end{lemma}

The stabilizer action described in \cref{lem:stabilizer-action} will be used to
implement declassification of termination channels in our denotational
semantics.

\begin{lemma}\label{lem:stab-action-preserves-ret}
  The stabilizer action preserves terminating computations in the sense that $\phi\parallel\Sub{\Lift{A}} u = u$ for $\phi:\Stab{A}$ and terminating $u:\Lift{A}$.
\end{lemma}

\begin{proof}
  We observe that $\phi\lor \top = \top$, hence for terminating $a$ we have
  $\phi\parallel\Sub{\Lift{A}} a = a$.
\end{proof}

\subsection{The denotational semantics}\label{sec:interpretation}

We now define an algebra for the theory $\TCat$ in $\Sh{\CmpTop}$;
the initial prefix of this algebra is standard:
\begin{gather*}
  \begin{mathblock}
    \bbrk{\Tp\Pos} = \DDPos\\
    \bbrk{\Tp\Neg} = \DDNeg\\
    \bbrk{\TpU}\,X = X\\
    \bbrk{\TpF}\,A= \Lift{A}\\
    \bbrk{\TmRet}\,a = a\\
    \bbrk{\TmBind}\,m\,f = f^\sharp\,m\\
    \bbrk{\TmFix}\,f = \Con{fix}\,{f}\\
    \bbrk{\TpFn}\,A\,X = A\Rightarrow X\\
    \bbrk{\TpFn.\Tm} = \MathComment{canonical}\\
  \end{mathblock}
  \begin{mathblock}
    \bbrk{\TpProd}\,A\,B = A\times B\\
    \bbrk{\TpProd.\Tm} = \MathComment{canonical}\\
    \bbrk{\TpUnit} = \ObjTerm{\DDPos}\\
    \bbrk{\TpUnit.\Tm} = \MathComment{canonical}\\
    \bbrk{\TpSum}\,A\,B = A+B\\
    \bbrk{\TmInl}\, a = \TmInl\,a\\
    \bbrk{\TmInr}\, a = \TmInr\,a\\
    \bbrk{\TmCase}\, u\,f\,g =
      \begin{cases}
        f\prn{x} & \text{if}\ u = \TmInl\,x\\
        g\prn{x} & \text{if}\ u = \TmInr\,x
      \end{cases}
  \end{mathblock}
\end{gather*}

Note that the coproduct $A+B$ above is computed in the category of predomains\footnote{Any reflective subcategory of a cocomplete category is cocomplete: first compute the colimit in the outer category, and then apply the reflection.} and need not be preserved by the embedding into $\Sh{\CmpTop}$. We next add the security levels and the sealing modality, interpreted as the
pushout of predomains $\ClMod{\LvlPol{l}}{A}$, again computed in the category of predomains.
We define the unsealing operator for $B :
\bbrk{\SealedTp{l}\Pos}$ using the universal property of the pushout.
\[
  \begin{mathblock}
    \bbrk{\LvlPol{l}} = \LvlPol{l} = \LvlMapCmp^*\Yo[\LVL]{l}\\
    \bbrk{\TpSeal{l}}\, A = \ClMod{\LvlPol{l}}{A}\\
    \bbrk{\TmSeal{l}}\, a = \ClIntro{\LvlPol{l}}\,{a}\\
  \end{mathblock}
  \qquad
  \begin{mathblock}
    \bbrk{\TmUnseal{l}}\,u\,f =
    \\
    \quad
    \begin{cases}
      f{x} & \text{if}\ u = \ClIntro{\LvlPol{l}}\,x\\
      \SlPt & \text{if}\ u = \SlPt
    \end{cases}
  \end{mathblock}
\]

\begin{observation}\label{obs:seal-universal-property}
  Morphisms $\Mor{\ClMod{\LvlPol{l}}{A}}{B}$ are in bijective correspondence
  with morphisms $\Mor{A}{B}$ that restricts to a \emph{weakly constant} function
  under $\LvlPol{l}$.
\end{observation}

We may now interpret the termination declassification operation.
Fixing a sealed type $A : \bbrk{\SealedTp{l}\Pos}$, we must define the dotted lift
below using the universal property of the pushout and the action of the
stabilizer of $A$ on $\Lift{A}$, noting that $\LvlPol{l}\in\Stab{A}$ by
assumption:
\begin{gather*}
  \begin{tikzpicture}[diagram,baseline=(nw.base)]
    \node (nw) {$A$};
    \node (ne) [right = 2.5cm of nw] {$\Lift{A}$};
    \node (sw) [below = of nw] {$\ClMod{\LvlPol{l}}\Lift{A}$};
    \draw[>->] (nw) to node [above] {$\eta_A$} (ne);
    \draw[->] (nw) to node [left] {$\ClIntro{\LvlPol{l}}\circ\eta_A$} (sw);
    \draw[->,exists] (sw) to node [sloped,below] {$\bbrk{\TmTdcl{l}}$} (ne);
  \end{tikzpicture}
  \qquad
  \begin{mathblock}
    \bbrk{\TmTdcl{l}}\, u =\\
    \quad
    \begin{cases}
      \LvlPol{l}\parallel\Sub{\Lift{A}} x & \text{if } u = \ClIntro{\LvlPol{l}}\,x\\
      \LvlPol{l}\parallel\Sub{\Lift{A}}\bot & \text{if } u = \SlPt
    \end{cases}
  \end{mathblock}
\end{gather*}

To see that the above is well-defined, we observe that under $\LvlPol{l}$
both branches return the (unique) computation whose termination support is
$\LvlPol{l}$. With this definition, the required computation rule holds by
virtue of \cref{lem:stab-action-preserves-ret}.

\subsection{Noninterference in the denotational semantics}

\begin{definition}
  A function $\Mor[u]{A}{B}$ is called \DefEmph{weakly
  constant}~\citep{kraus-escardo-coquand-altenkirch:2017} if for all $x,y:A$ we
  have $u\,x = u\,y$. A partial function $\Mor[u]{A}{\Lift{B}}$ is called
  \DefEmph{partially constant} if for all $x,y: A$ such that
  $\IsDefd{u\,x}\land\IsDefd{u\,y}$, we have $u\,x = u\,y$.
\end{definition}

For the following, let $l\in\LVL$ be a security level.

\begin{restatable}{lemma}{LemConstancy}\label{lem:constancy}
  Let $A$ be a $\LvlPol{l}$-sealed predomain and let $B$ be a
  $\LvlPol{l}$-transparent predomain; then (1) any function $\Mor{A}{B}$ is
  weakly constant, and (2) any partial function $\Mor{A}{\Lift{B}}$ is
  partially constant.
\end{restatable}

The following lemma follows from \cref{axiom:finset-transparent-predomain}.

\begin{restatable}{lemma}{LemBoolModal}\label{lem:bool-modal}
  The predomain $\bbrk{\TpBool}$ is $\LvlPol{l}$-transparent.
\end{restatable}

In order for \cref{lem:constancy} to have any import as far as the equational
theory is concerned, we must establish computational adequacy. This is the
topic of \cref{sec:adequacy}.

\section{Adequacy of the denotational semantics}\label{sec:adequacy}

We must argue that the denotational semantics agrees with the theory as far as
convergence and return values is concerned. We do so using a Plotkin-style
logical relations argument, phrased in the language of Synthetic Tait
Computability~\citep{sterling:2021:thesis,sterling-harper:2021,sterling-angiuli:2021}.

\subsection{Synthetic Tait computability of formal approximation}\label{sec:logical-relation}

In this section we will work abstractly with a Grothendieck topos $\GlTop$
satisfying several axioms that will make it support a Kripke logical relation
for adequacy.

\begin{notation}
  For each universe $\UU\in\Sh{\GlTop}$ there is a type $\ALG{\UU}$ of
  internal $\TCat$-algebras whose type components are valued in $\UU$.
  $\ALG{\UU}$ is a dependent record containing a field for every constant in
  the signature by which we generated $\TCat$.
  Assuming enough universes, functors $\Mor{\TCat}{\Sl*{\Sh{\GlTop}}{E}}$
  correspond up to isomorphism to morphisms $\Mor{E}{\ALG{\UU}}$. This is the
  relationship between the internal language and the \emph{functorial
  semantics} \`a la Lawvere~\citep{lawvere:thesis:reprint}.
\end{notation}

\begin{restatable}{stcaxiom}{AxPhases}\label{axiom:syn-cmp-phases}
  There are two disjoint propositions $\SynOpn,\CmpOpn\in\Opns{\GlTop}$
  such that $\SynOpn\land\CmpOpn=\bot$. We will refer to these as the
  \DefEmph{syntactic} and \DefEmph{computational phases} respectively. We will write
  $\BaseOpn=\SynOpn\lor\CmpOpn$ for the disjoint union of the two phases.
\end{restatable}

\begin{restatable}{stcaxiom}{AxGenericModel}\label{axiom:syn-alg}
  Within the syntactic phase, there exists a $\TCat$-algebra $\SynAlg :
  \ALG{\Sl{\UU}{\SynOpn}}$ such that the corresponding functor
  $\Mor{\TCat}{\Sl*{\Sh{\GlTop}}{\SynOpn}}$ is fully faithful.
\end{restatable}

\begin{restatable}{stcaxiom}{AxSdtInCmpPhase}\label{axiom:sdt-in-cmp-phase}
  Within the computational phase, the axioms of $\LVL$-indexed synthetic domain
  theory
  (\cref{axiom:sigma-finite-joins,axiom:sigma-predomain,axiom:omega-inductive,axiom:lvls,axiom:finset-transparent-predomain})
  are satisfied.
\end{restatable}

As a consequence of \cref{axiom:sdt-in-cmp-phase}, we have a
\emph{computational} $\TCat$-algebra $\CmpAlg:\ALG{\Sl{\UU}{\CmpOpn}}$ given by
the constructions of \cref{sec:interpretation}.  Gluing together the two models
$\SynAlg,\CmpAlg$ we see that $\Sl{\GlTop}{\BaseOpn}$ supports a model
$\BaseAlg = \brk{\SynOpn\hookrightarrow \SynAlg,\CmpOpn\hookrightarrow
\CmpAlg}$ of $\TCat$.
The final \cref{axiom:lvls-under-closed-immersion} above is needed in the
approximation structure of $\TmTdcl{l}$.

\begin{restatable}{stcaxiom}{AxLvlsUnderClosedImmersion}\label{axiom:lvls-under-closed-immersion}
  For each $l\in\LVL$ we have ${\CmpAlg.\LvlPol{l}} \leq \ClMod{\BaseOpn}{\SynAlg.\LvlPol{l}}$.
\end{restatable}

\begin{theorem}
  There exists a topos $\GlTop$ satisfying
  \cref{axiom:syn-cmp-phases,axiom:syn-alg,axiom:sdt-in-cmp-phase,axiom:lvls-under-closed-immersion}
  containing open subtopoi $\Sl{\GlTop}{\SynOpn} = \PrTop{\TCat}$ and $\Sl{\GlTop}{\CmpOpn} =
  \CmpTop$ such that the complementary closed subtopos is $\ClSubcat{\GlTop}{\BaseOpn} = \LvlTop$.
\end{theorem}

\begin{proof}
  We may construct a topos using a variant of the Artin gluing construction of
  \citet[Sterling and Harper]{sterling-harper:2021}, which we detail in our
  \PhraseAppendix.
\end{proof}

By \cref{axiom:syn-cmp-phases,axiom:syn-alg}, any such topos $\GlTop$ supports a
model of the \DefEmph{synthetic Tait computability} of
\citet[Sterling and Harper]{sterling-harper:2021,sterling:2021:thesis}. In the internal language of
$\Sh{\GlTop}$, the phase $\BaseOpn$ induces a pair of complementary
transparency/open and sealing/closed modalities that can be used to
synthetically construct formal approximation relations in the sense of Plotkin
between computational objects and syntactical objects.
Viewing an object $E\in\Sh{\GlTop}$ as a
family $x:\OpMod{\CmpOpn}{E},x':\OpMod{\SynOpn}{E}\vdash
\Compr{E}{\CmpOpn\hookrightarrow x, \SynOpn\hookrightarrow x'}$ of
$\BaseOpn$-sealed types over the $\BaseOpn$-transparent type
$\prn{\OpMod{\BaseOpn}{E}} \cong \prn{\prn{\OpMod{\CmpOpn}{E}}\times \prn{\OpMod{\SynOpn}{E}}}$, we may
think of $E$ as a \emph{proof-relevant} formal approximation relation between
its computational and syntactic parts, which we might term a ``formal
approximation structure''.

\begin{notation}[Extension types]
  We recall \DefEmph{extension types} from \citet[Riehl and Shulman]{riehl-shulman:2017}. Given a
  proposition $\phi:\Omega$ and a partial element $e : \OpMod{\phi}{E}$, we will write
  $\Ext{E}{\phi}{e}$ for the collection of elements of $E$ that restrict to $e$
  under $\phi$, \ie the subobject $\Mor|>->|{\Compr{x : E}{\OpMod{\phi} \prn{x
  = e}}}{E}$.  Note that $\Ext{E}{\phi}{e}$ is always
  $\phi$-sealed, since it becomes the singleton type $\brc{e}$ under $\phi$.
\end{notation}

Each universe $\UU$ of $\Sh{\GlTop}$ satisfies a remarkable
\emph{strictification} property with respect to any proposition $\phi:\Omega$
that allows one to construct codes for dependent sums of families of
$\phi$-sealed types over a $\phi$-transparent type in such a way that they
restrict \emph{exactly} to the $\phi$-transparent part under $\phi$. This refinement of dependent sums is called a \DefEmph{strict glue type}:\footnote{In presheaves, the universes of
\citet[Hofmann and Streicher]{hofmann-streicher:1997,streicher:2005} satisfy this property directly;
for sheaves, there is an alternative transfinite construction of
universes enjoying this property~\citep{gratzer-shulman-sterling:2022:universes}. Our presentation in terms of transparency and sealing is an equivalent reformulation of the strictness property identified by several authors in the context of the semantics of homotopy type theory~\citep{kapulkin-lumsdaine:2021,streicher:2014:simplicial,shulman:2015:elegant,cchm:2017,orton-pitts:2016,bbcgsv:2016,shulman:2019,awodey:2021:qms}.}
\begin{mathpar}
  \inferrule[strict glue types]{
    A:\OpMod{\phi}{\UU}\\
    B : \prn{\OpMod{\prn{z:\phi}}A\,z} \to \UU\\
    \forall
    x.\
    \Con{isSealed}\Sub{\phi}\,\prn{B,x}
  }{
    \GlueFam{\phi}{x:A}{B\,x} : \Ext{\UU}{z:\phi}{A\,z}\\\\
    \Con{glue}_\phi : \Ext{\prn{\prn{x:\OpMod{\prn{z:\phi}}{A\,z}}\times B\,x} \cong
    \GlueFam{\phi}{x:A}{B\,x}}{\phi}{\pi_1}
  }
\end{mathpar}

\begin{notation}[Strict glue types]
  We impose two notations assuming $A,B$ as above.
  Given $a:\OpMod{\prn{z:\phi}}{A\,z}$ and $b : B\,a$, we write $\GlueEl{b}{\phi}{a}$ for $\Con{glue}_\phi\prn{a,b}$.
  Given $g : \GlueFam{\phi}{x:A}{B\,x}$, we write $\Unglue{\phi}g : B\,g$ for the element $\pi_2\,\prn{\Inv{\Con{glue}_\phi}\,g}$.
\end{notation}

\ifextendedversion

Using the strict glue types it is possible to define very strict universes
$\Sl{\UU}{\phi},\ClSubcat{\UU}{\phi}$ of $\phi$-transparent and $\phi$-sealed
types respectively which are themselves $\phi$-transparent and $\phi$-sealed in
the next universe as in \citet[\S3.6 of Sterling's dissertation]{sterling:2021:thesis}:
\[
  \begin{mathblock}
    \Sl{\UU}{\phi} : \Ext{\VV}{\phi}{\UU}\\
    \ClSubcat{\UU}{\phi} : \Ext{\VV}{\phi}{\ObjTerm}
  \end{mathblock}
  \quad
  \begin{mathblock}
    \prn{\OpMod{\phi}{-}} : \Ext{\UU\to\Sl{\UU}{\phi}}{\phi}{\lambda A.A}\\
    \prn{\ClMod{\phi}{-}} : \UU\to\ClSubcat{\UU}{\phi}
  \end{mathblock}
\]

The transparent subuniverse $\Sl{\UU}{\phi}$ is canonically isomorphic to
$\OpMod{\phi}{\UU}$; given an element $A:\Sl{\UU}{\phi}$, elements of $A$ are the same
as partial elements $\prn{z:\phi}\to A\,z$ under the former identification. In
our notation, we suppress these identifications \emph{as well as} the
introduction and elimination form for these $\phi$-partial elements. The sealed
subuniverse $\ClSubcat{\UU}{\phi}$ is canonically isomorphic to
$\Ext{\UU}{\phi}{\ObjTerm}$. The modal universes
$\Sl{\UU}{\phi},\ClSubcat{\UU}{\phi}$ serve as \emph{(weak) generic objects}
in the sense of \citet[Jacobs]{jacobs:1999}
for the full subfibrations of $\FibMor[\Cod\Sub{\Sh{\GlTop}}]{\Sh{\GlTop}^\to}{\Sh{\GlTop}}$ spanned by fiberwise
$\UU$-small $\phi$-transparent and $\phi$-sealed families of types
respectively.

\fi

\begin{notation}
  Let $E$ be a type in $\Sh{\GlTop}$ and fix elements $e : \OpMod{\CmpOpn}{E}$ and
  $e' : \OpMod{\SynOpn}{E}$ of the computational and syntactical parts of $E$
  respectively; we will write $e \lhd_E e'$, pronounced ``$e$ formally approximates
  $e'$'', for the extension type
  $\Compr{E}{\CmpOpn\hookrightarrow e, \SynOpn\hookrightarrow e'}$.
\end{notation}

This is the connection between synthetic Tait
computability and analytic logical relations; the open parts of an object
correspond to the \emph{subjects} of a logical relation and the closed parts of
an object correspond to the evidence of that relation.

\begin{definition}[Formal approximation relations]
  A type $E$ is called a \DefEmph{formal approximation relation} when for any
  $\BaseOpn$-point $e:\OpMod{\BaseOpn}{E}$, the extension type
  $\Ext{E}{\BaseOpn}{e}$ is a proposition, \ie any two elements of $e\lhd_E e$
  are equal.
\end{definition}

We will write $\Rel{\UU}\subseteq\UU$ for the subuniverse of formal approximation
relations.

\begin{definition}[Admissible formal approximation relations]
  Let $E$ be a formal approximation relation such that $\OpMod{\CmpOpn}{E}$ is
  a predomain equipped with an $\Lift$-algebra structure. We say that $E$
  is \DefEmph{admissible} at $x:\OpMod{\SynOpn}{E}$ when the subobject
  $\Ext{E}{\SynOpn}{x}\subseteq\OpMod{\CmpOpn}{E}$ is admissible in the sense
  of synthetic domain theory, \ie contains $\bot$ and is closed
  under formal suprema of formal $\ShChain$-chains. We say that $E$ is
  admissible when it is admissible at every such $x$.
\end{definition}

\begin{restatable}[Scott induction]{lemma}{LemScottInduction}\label{lem:scott-induction}
  Let $X$ be a formal approximation relation such that
  $\OpMod{\CmpOpn}{X}$ is a domain. Let $f : X\to X$ be an endofunction on $X$ and let
  $x:\OpMod{\SynOpn}{X}$ be a syntactical fixed point of $f$ in the sense that
  $\OpMod{\SynOpn}{\prn{x=f\,{x}}}$; if $X$ is admissible at $x$, then we have $\Con{fix}\,f\lhd_X x$.
\end{restatable}

Our goal can be rephrased now in the internal language; choosing a universe
$\VV\supset \UU$, we wish to define a suitable $\VV$-valued algebra
$\Alg\in\ALG{\VV}$ that restricts under $\BaseOpn$ to $\BaseAlg$, \ie an element
$\Alg\in\Ext{\ALG{\VV}}{\BaseOpn}{\BaseAlg}$. This can be done quite elegantly in
the internal language of $\Sh{\GlTop}$, \ie the \emph{synthetic Tait
computability of formal approximation structures}.
The high-level structure of our model construction is summarized as follows:
\begin{quote}
  We interpret value types as \DefEmph{formal approximation structures} over a
  syntactic value type and a predomain; we interpret computation types as
  \DefEmph{admissible formal approximation relations} between a syntactic
  computation type and a domain.
\end{quote}

To make this precise, we will define $\Alg.\Tp\Pos \in
\Ext{\VV}{\BaseOpn}{\BaseAlg.\Tp\Pos}$ as the collection of types that restrict
to an element of $\SynAlg.\Tp\Pos$ in the syntactic phase and to an element of
$\CmpAlg.\Tp\Pos = \DDPos$ in the computational phase. This is achieved using
strict gluing:
\[
  \begin{mathblock}
    \Alg.\Tp\Pos = \GlueFam{\BaseOpn}{A:\BaseAlg.\Tp\Pos}{\Ext{\UU}{\BaseOpn}{\BaseAlg.\Tm\,A}}
  \end{mathblock}
  \quad
  \begin{mathblock}
    \Alg.\Tm = \Unglue{\BaseOpn}
  \end{mathblock}
\]

The above is well-defined because $\BaseAlg.\Tp\Pos$ is $\BaseOpn$-transparent
and $\Ext{\UU}{\BaseOpn}{\BaseAlg.\Tm\,A}$ is $\BaseOpn$-sealed. We also have
$\OpMod{\SynOpn}{\Alg.\Tp\Pos} = \SynAlg.\Tp\Pos$ and
$\OpMod{\CmpOpn}{\Alg.\Tp\Pos} = \DDPos$. Next we define the formal approximation structure of
computation types:
\[
  \begin{mathblock}
    \Alg.\Tp\Neg =
    \GlueFam{\BaseOpn}{
      X:\BaseAlg.\Tp\Neg
    }{
      \Compr{X' : \Ext{\Rel{\UU}}{\BaseOpn}{\BaseAlg.\Tm\,\prn{\BaseAlg.\TpU\, X}}}{
        X'\ \text{is admissible}
      }
    }
  \end{mathblock}
\]

To see that the above is well-defined, we must check that the family component
of the gluing is pointwise $\BaseOpn$-sealed, which follows because the
property of being admissible is $\BaseOpn$-sealed. To see that this is the
case, we observe that it is obviously $\SynOpn$-sealed and also (less
obviously) $\CmpOpn$-sealed: under $\CmpOpn$, $X'$ restricts to the ``total''
predicate on $X$ which is always admissible.
To define the thunking connective, we simply forget that a given admissible
approximation relation was admissible:
$\Alg.\TpU\, X = \GlueEl{ \Unglue{\BaseOpn}{X} }{\BaseOpn}{\BaseAlg.\TpU\,X}$.
To interpret free computation types, we proceed in two steps; first we define
the formal approximation relation as an element of $\Rel{\UU}$ and then we glue
it onto syntax and semantics.
\[
  \begin{mathblock}
    \brk{\TpF}\,A =
      \GlueFam{\BaseOpn}{
        u : \BaseAlg.\prn{\TpU\TpF}\,A
      }{
        \prn{\OpMod{\CmpOpn}\IsDefd{u}}
        \Rightarrow
        \ClMod{\BaseOpn}{
          \exists a:A.
          \OpMod{\BaseOpn}{
            u = \BaseAlg.\TmRet\,a
          }
        }
      }
    \\%
    \Alg.\TpF\,A =
    \GlueEl{
      \brk{\TpF}\,A
    }{\BaseOpn}{\BaseAlg.\TpF}
  \end{mathblock}
\]

In simpler language, we have $u \lhd\Sub{\brk{\TpF}\,A} v$ if and only if $v$
terminates syntactically whenever $u$ terminates such that the value of $u$
formally approximates the value of $v$. This is the standard clause for lifting
in an adequacy proof, phrased in synthetic Tait computability.
; the use of the sealing modality is an artifact of synthetic Tait
computability, ensuring that the relation is pointwise $\BaseOpn$-sealed.
The $\TmRet,\TmBind$ operations are easily shown to preserve the formal
approximation relations.  The construction of formal approximation structures
for product and function spaces is likewise trivial.  Using Scott induction
(\cref{lem:scott-induction}) we can show that fixed points also lie in the
formal approximation relations; we elide the details.
Next we deal with the information flow constructs, starting by interpreting
each security policy $\Alg.\LvlPol{l}$ as $\BaseAlg.\LvlPol{l}$.
The sealing modality is interpreted below:
\[
  \begin{mathblock}
    \brk{\TpSeal{l}}\,A =
    \GlueFam{\BaseOpn}{
      u : \BaseAlg.\TpSeal{l}A
    }{
      \ClMod{\BaseOpn}{
        \ClMod{\Alg.\LvlPol{l}}{
          \Compr{
            a:A
          }{
            \OpMod{\BaseOpn}{
              u = \BaseAlg.\TmSeal{l}a
            }
          }
        }
      }
    }
    \\%
    \Alg.\TpSeal{l}A =
    \GlueEl{
      \brk{\TpSeal{l}}\,A
    }{\BaseOpn}{\BaseAlg.\TpSeal{l}}
  \end{mathblock}
\]

\begin{restatable}[Fundamental theorem of logical relations]{theorem}{ThmFTLR}
  The preceding constructions arrange into an algebra $\Alg \in \Ext{\ALG{\VV}}{\BaseOpn}{\BaseAlg}$.
\end{restatable}

\subsection{Adequacy and syntactic noninterference results}

The following definitions and results in this section are global rather
than internal. We may immediately read off from the logical relation of
\cref{sec:logical-relation} a few important properties relating value terms and
their denotations. The results of this section depend heavily on the assumption
that the functor $\EmbMor{\TCat}{\Sl*{\Sh{\GlTop}}{\SynOpn}}$ is fully faithful
(\cref{axiom:syn-alg}).

\NewDocumentCommand\Converges{m}{{#1}{\Downarrow}}
\NewDocumentCommand\Diverges{m}{{#1}{\Uparrow}}

\begin{restatable}[Value adequacy]{theorem}{ThmValAdequacy}\label{thm:value-adequacy}
  For any closed values $\Mor[u,v]{\ObjTerm{\TCat}}{\TpBool}$, we have $\bbrk{u} =
  \bbrk{v}$ if and only if $u\equiv\Sub{\TpBool}v$; moreover we have either $u\equiv\Sub{\TpBool}\TmTt$ or $u\equiv\Sub{\TpBool}\TmFf$.
\end{restatable}

Let $\Mor[u]{\ObjTerm{\TCat}}{\TpU\TpF{A}}$ be a closed computation.

\begin{definition}[Convergence and divergence]
  We say that $u$
  \DefEmph{converges} when there exists $\Mor[a]{\ObjTerm{\TCat}}{A}$ such that $u
  = \TmRet\,a$. Conversely, we say that $u$ \DefEmph{diverges} when there does not
  exist such an $a$. We will write $\Converges{u}$ to mean that $u$ converges,
  and $\Diverges{u}$ to mean that $u$ diverges.
\end{definition}

\begin{restatable}[Computational adequacy]{theorem}{ThmCmpAdequacy}\label{thm:cmp-adequacy}
  The computation $u$ converges iff $\IsDefd{\bbrk{u}} = \top$.
\end{restatable}

\begin{restatable}[Termination-insensitive noninterference]{theorem}{ThmTini}\label{thm:tini}
  Let $A$ be a syntactic type such that $\SealedBelow{l}\,A$ holds; fix a term
  $\Mor[c]{A}{\TpU\TpF\,\TpBool}$. Then for all
  $\Mor[x,y]{\ObjTerm{\TCat}}{A}$ such that $\Converges{c\,x}$ and
  $\Converges{c\,y}$, we have $c\,x \equiv\Sub{\TpU\TpF\,\TpBool} c\,y$.
\end{restatable}

We give an example of a program whose termination behavior hinges on a
classified bit to demonstrate that our noninterference result is non-trivial.

\begin{example}
  There exists a $\LvlPol{l}$-sealed type $A$ and a term
  $\Mor[c]{A}{\TpU\TpF\,\TpUnit}$ such that for some $\Mor[x,y]{\ObjTerm{\TCat}}{A}$ we have
  $\Converges{c\,x}$ and yet $\Diverges{c\,y}$.%
\end{example}

\begin{proof}
  Choose $A\coloneqq\TpSeal{l}\,\TpBool$ and consider the following terms:
  \[
    \begin{mathblock}
      \top \coloneqq \TmRet\,\prn{}\quad
      \bot \coloneqq \TmFix\,\prn{\lambda z.z}\quad
      x \coloneqq \TmSeal{l}\,\TmTt\quad
      y \coloneqq \TmSeal{l}\,\TmFf\\
      c \coloneqq \lambda u.\,
      \TmTdcl{l}\,
      \prn{
        \TmUnseal{l}\, u\,
        \prn{
          \lambda b.\,
          \TmSeal{l}\,\prn{
            \TmIf\, b\, \top\, \bot
          }
        }
      }
    \end{mathblock}
  \]

  We then have $c\,x \equiv\Sub{\TpU\TpF\,\TpUnit} \top$ and therefore
  $\Converges{c\,x}$. On the other hand, we have $c\,y
  \equiv\Sub{\TpU\TpF\,\TpUnit} \TmTdcl{l}\,\prn{\TmSeal{l}\, \bot}$; executing
  the denotational semantics, we have $\IsDefd{\bbrk{c\,y}} = \LvlPol{l}$.
  From the full and faithfulness assumption of \cref{axiom:lvls}, we know that
  $\LvlPol{l}$ is not globally equal to $\top$; hence we conclude from
  \cref{thm:cmp-adequacy} that $\Diverges{c\,y}$.
\end{proof}
 
\bibliography{refs-stripped}

\ifextendedversion
\clearpage
\markboth{Technical Appendix}{}

\begin{appendices}
\section*{How to read the technical appendix}

We divide our technical appendix into two parts. \cref{apx:synthetic} provides
proofs of the main results of this paper that depend only on the axiomatics we
have imposed for synthetic domain theory and synthetic Tait computability. To
verify that these axioms are substantiated, we construct concrete models of
synthetic domain theory and synthetic Tait computability in
\cref{apx:analytic}.

\begin{figure}
  \footnotesize
  \[
    \begin{array}{l}
      \begin{array}{lll}
        \begin{mathblock}
          \RuleBlockJudgmental
        \end{mathblock}
        &
        \begin{mathblock}
          \RuleBlockRetBind\mathbreak
          \RuleBlockBindEqns
        \end{mathblock}
        &
        \begin{mathblock}
          \RuleBlockFn\mathbreak
          \RuleBlockProd\mathbreak
          \RuleBlockUnit
        \end{mathblock}
      \end{array}
      \\\\
      \begin{array}{lll}
        \begin{mathblock}
          \RuleBlockSum\mathbreak
          \RuleBlockCase
        \end{mathblock}
        &
        \begin{mathblock}
          \begin{mathblock}
          \RuleBlockLvl
          \end{mathblock}
          \mathbreak
          \begin{mathblock}
          \RuleBlockIsSealed\mathbreak
          \SealedTp{l}\Pos\coloneqq \Compr{A : \Tp\Pos}{\SealedBelow{l}\,A}\\
          \end{mathblock}
        \end{mathblock}
      \end{array}
      \\\\
      \begin{array}{ll}
        \begin{mathblock}
          \RuleBlockSeal\mathbreak
          \RuleBlockUnseal
        \end{mathblock}
        &
        \begin{mathblock}
          \RuleBlockFix\mathbreak
          \RuleBlockTdcl
        \end{mathblock}
      \end{array}
    \end{array}
  \]
  \caption{The LF presentation of the theory $\TCat$.}
  \label{fig:lf-sig}
\end{figure}

\section{Synthetic results}\label{apx:synthetic}

\subsection{Denotational semantics in synthetic domain theory}

In this section we assume all the axioms of synthetic domain theory.

\begin{lemma}
  The predomains form a \DefEmph{reflective exponential ideal} of $\Sh{\CmpTop}$.
\end{lemma}

\begin{proof}
  In a locally presentable category, such as any Grothendieck topos, any
  orthogonal subcategory is reflective.
\end{proof}

\subsubsection{Denotational noninterference}

\LemConstancy*

\begin{proof}
  Fix $\LvlPol{l}$-sealed $A$ and $\LvlPol{l}$-transparent $B$.  For
  (1), we fix a function $\Mor[u]{A}{B}$ to check that
  for any $x,y:A$ we have $u\,x= u\,y$. Because $B$ is $\LvlPol{l}$-transparent,
  we may assume $\LvlPol{l}=\top$; but then $x=y$ and hence $u\,x = u\,y$.
  For (2), we fix a partial function
  $\Mor[u]{A}{B_\bot}$ along with $x,y:A$ such that $u$ is defined on both $x$
  and $y$; we must check that $u\,x = u\,y : B$. Because $B$ is
  $\LvlPol{l}$-transparent, we may again assume $\LvlPol{l}=\top$; but $A$ is
  $\LvlPol{l}$-sealed and hence under this assumption, we have $x = y$ and
  therefore $u\,x = u\,y$.
\end{proof}

\LemBoolModal*

\begin{proof}
  $\bbrk{\TpBool}$ is the binary coproduct $\Two{\PREDOM}$ of singleton
  predomains, which by \cref{axiom:finset-transparent-predomain} is just the constant sheaf
  $\Two{\Sh{\CmpTop}} \cong \CmpTop^*\brc{0,1}$.
\end{proof}

\subsection{Adequacy via synthetic Tait computability}

\LemScottInduction*

\begin{proof}
  We employ the fact that admissible approximation relations are closed under
  bottom and formal suprema of $\ShChain$-chains.
\end{proof}

\subsubsection{Universe of value types}

\[
  \begin{mathblock}
    \Alg.\Tp\Pos:\Ext{\VV}{\BaseOpn}{\BaseAlg.\Tp\Pos}\\
    \Alg.\Tp\Pos = \GlueFam{\BaseOpn}{A:\BaseAlg.\Tp\Pos}{\Ext{\UU}{\BaseOpn}{\BaseAlg.\Tm\,A}}
    \\[6pt]
    \Alg.\Tm : \Ext{\Alg.\Tp\Pos\to\VV}{\BaseOpn}{\BaseAlg.\Tm}\\
    \Alg.\Tm = \Unglue{\BaseOpn}
  \end{mathblock}
\]

\subsubsection{Universe of computation types}

\[
  \begin{mathblock}
    \Alg.\Tp\Neg : \Ext{\VV}{\BaseOpn}{\BaseAlg.\Tp\Neg}\\
    \Alg.\Tp\Neg =
    \GlueFam{\BaseOpn}{
      X:\BaseAlg.\Tp\Neg
    }{
      \Compr{X' : \Ext{\Rel{\UU}}{\BaseOpn}{\BaseAlg.\Tm\,\prn{\BaseAlg.\TpU\, X}}}{
        X'\ \text{is admissible}
      }
    }
    \\[6pt]
    \Alg.\TpU : \Ext{\Alg.\Tp\Neg\to\Alg.\Tp\Pos}{\BaseOpn}{\BaseAlg.\TpU}\\
    \Alg.\TpU\, X =
    \GlueEl{
      \Unglue{\BaseOpn}{X}
    }{\BaseOpn}{\BaseAlg.\TpU\,X}
  \end{mathblock}
\]

\subsubsection{Free computation types}

First we define the underlying formal approximation relation.
\[
  \begin{mathblock}
    \brk{\TpF} : \Ext{\Alg.\Tp\Pos\to\Rel{\UU}}{\BaseOpn}{\BaseAlg.\prn{\Tm\circ\TpU\circ\TpF}}\\
    \brk{\TpF}\,A =
      \GlueFam{\BaseOpn}{
        u : \BaseAlg.\prn{\TpU\TpF}\,A
      }{
        \prn{\OpMod{\CmpOpn}\IsDefd{u}}
        \Rightarrow
        \ClMod{\BaseOpn}{
          \exists a:A.
          \OpMod{\BaseOpn}\prn{
            u = \BaseAlg.\TmRet\,a
          }
        }
      }
  \end{mathblock}
\]

\begin{computation}
  It is useful to characterize the relation in more familiar notation:
  \[
    u \lhd\Sub{\brk{\TpF}\,A} x \Longleftrightarrow
    \IsDefd{u}
    \Rightarrow
    \ClMod{\BaseOpn}{
      \exists u',x'.
      \OpMod{\CmpOpn}{\prn{u = \eta\Sub{A} u'}}
      \land
      \OpMod{\SynOpn}{\prn{x=\SynAlg.\TmRet\,x'}}
      \land
      u' \lhd\Sub{A} x'
    }
  \]
\end{computation}

\begin{lemma}\label{lem:F-admissible}
  Each $\brk{\TpF}\,A$ is admissible.
\end{lemma}

\begin{proof}
  We fix a syntactic point $x:\SynAlg.\TpU\TpF{A}$ to check that each
  $\Ext{\brk{\TpF}\,A}{\SynOpn}{x}\subseteq\Lift{A}$ is an admissible
  subdomain.
  \begin{enumerate}

    \item It is immediate that $\bot\Sub{\Lift{A}} \lhd\Sub{\brk{\TpF}\,A}
      x$, since $\IsDefd{\bot\Sub{\Lift{A}}} = \bot$.

    \item Next we fix an $\ShChain$-chain of formal approximations
      $u_i\lhd\Sub{\brk{\TpF}\,A}x$ to check that
      $\DLub{i}{u_i}\lhd\Sub{\brk{\TpF}\,A}x$. Because
      $\IsDefd{\DLub{i}{u_i}} = \DLub{i}{\IsDefd{u_i}}$, we assume that
      $\IsDefd{u_i}$ for some $i$ to check that away from $\BaseOpn$ there
      exists an element $a:A$ such that $x$ restricts under $\BaseOpn$ to
      $\SynAlg.\TmRet\,a$ and under $\CmpOpn$ we have $\DLub{i}{u_i} =
      \eta\Sub{A}a$. Because $\IsDefd{u_i}$ and therefore the value of $u_i$ is
      the same as the value of $\DLub{i}{u_i}$, we are done.
      \qedhere

  \end{enumerate}
\end{proof}

By \cref{lem:F-admissible} we are justified in specifying the code for the free
computation type:
\[
  \begin{mathblock}
    \Alg.\TpF : \Ext{\Alg.\Tp\Pos\to\Alg.\Tp\Neg}{\BaseOpn}{\BaseAlg.\TpF}\\
    \Alg.\TpF\,A =
    \GlueEl{
      \brk{\TpF}\,A
    }{\BaseOpn}{\BaseAlg.\TpF}
  \end{mathblock}
\]

\begin{lemma}
  If $u\lhd\Sub{A} u'$ then $\eta\Sub{A}u\lhd\Sub{\brk{A}\,A}\SynAlg.\TmRet\,{u'}$
\end{lemma}

\begin{proof}
  Immediate.
\end{proof}

\begin{lemma}
  If $u\lhd\Sub{\brk{F}\,A}u'$ and for all $a\lhd_A a'$ we have $ha\lhd\Sub{\Alg.\TpU X} h'a'$
  then we have $h^\sharp u \lhd\Sub{\Alg.\TpU X} \SynAlg.\TmBind\,u'\,h'$.
\end{lemma}

\begin{proof}
  Assume $\IsDefd{h^\sharp u}$, \ie there exists $a$ such that $u =
  \eta\Sub{A}a$ and $\IsDefd{ha}$. Therefore there exists $a'$ with $u' =
  \SynAlg.\TmRet\,a'$ and $a\lhd_A a'$.  In this case $h^\sharp u = ha$ and
  $\SynAlg.\TmBind\,u'\,h' = h'a'$, and we have $ha\lhd\Sub{\Alg.\TpU X} h'a'$ by
  assumption.
\end{proof}

\subsubsection{Function types}

First we define the underlying formal approximation relation.
\[
  \brk{\TpFn}\,A\,X =
  \GlueFam{\BaseOpn}{
    u:\BaseAlg.\TpU\,\prn{\BaseAlg.\TpFn\,A\,X}
  }{
    \Ext{
      A \to \Alg.\TpU X
    }{\BaseOpn}{
      \Inv{\BaseAlg.\TpFn.\Tm}\,u
    }
  }
\]

That the above is admissible follows immediately from the definition of formal
suprema in function spaces.
\[
  \begin{mathblock}
    \Alg.\TpFn : \Ext{\Alg.\Tp\Pos\to\Alg.\Tp\Neg\to\Alg.\Tp\Neg}{\BaseOpn}{\BaseAlg.\TpFn}\\
    \Alg.\TpFn\,A\,X = \GlueEl{\brk{\TpFn}\,A\,X}{\BaseOpn}{\BaseAlg.\TpFn}
  \end{mathblock}
\]

The canonical isomorphism $\Alg.\TpFn.\Tm$ is immediate:
\[
  \begin{mathblock}
    \Alg.\TpFn.\Tm : \Ext{\prn{A\to \Alg.\TpU{X}}\cong \Alg.\TpU\,\prn{\Alg.\TpFn\,A\,X}}{\BaseOpn}{\BaseAlg.\TpFn.\Tm}\\
    \Alg.\TpFn.\Tm\,u = \GlueEl{u}{\BaseOpn}{\BaseAlg.\TpFn.\Tm\,u}
  \end{mathblock}
\]

\subsubsection{Product types}

First we define the underlying formal approximation structure.
\[
  \brk{\TpProd}\,A\,B = \GlueFam{\BaseOpn}{u:\BaseAlg.\TpProd\,A\,B}{\Ext{A\times B}{\BaseOpn}{\Inv{\BaseAlg.\TpProd.\Tm}\,u}}
\]

The rest is defined as follows:
\[
  \begin{mathblock}
    \Alg.\TpProd : \Ext{\Alg.\Tp\Pos\to\Alg.\Tp\Pos\to\Alg.\Tp\Pos}{\BaseOpn}{\BaseAlg.\TpProd}\\
    \Alg.\TpProd\,A\,B = \GlueEl{\brk{\TpProd}\,A\,B}{\BaseOpn}{\BaseAlg.\TpProd}\\[6pt]
    \Alg.\TpProd.\Tm : \Ext{\prn{A\times B}\cong \Alg.\TpProd\,A\,B}{\BaseOpn}{\BaseAlg.\TpProd.\Tm}\\
    \Alg.\TpProd.\Tm\,u = \GlueEl{u}{\BaseOpn}{\BaseAlg.\TpProd.\Tm\,u}
  \end{mathblock}
\]

\subsubsection{Sum types}

First we define the underlying formal approximation structure.
\[
  \brk{\TpSum}\,A\,B =
  \GlueFam{\BaseOpn}{u:\BaseAlg.\TpSum\,A\,B}{
    \ClMod{\BaseOpn}{
      \prn{
        \Compr{x:A}{\OpMod{\BaseOpn} u=\BaseAlg.\TmInl\, x}
        +
        \Compr{x:B}{\OpMod{\BaseOpn} u=\BaseAlg.\TmInr\, x}
      }
    }
  }
\]

The type code and its constructors are defined as follows:
\[
  \begin{mathblock}
    \Alg.\TpSum : \Ext{\Alg.\Tp\Pos\to\Alg.\Tp\Pos\to\Alg.\Tp\Pos}{\BaseOpn}{\BaseAlg.\TpSum}\\
    \Alg.\TpSum\,A\,B = \GlueEl{\brk{\TpSum}\,A\,B}{\BaseOpn}{\BaseAlg.\TpSum}
    \\[6pt]
    \Alg.\TmInl : \Ext{A \to \Alg.\TpSum\,A\,B}{\BaseOpn}{\BaseAlg.\TmInl}\\
    \Alg.\TmInl\,a = \GlueEl{\TmInl\,a}{\BaseOpn}{\BaseAlg.\TmInl\,a}
    \\[6pt]
    \Alg.\TmInr : \Ext{B \to \Alg.\TpSum\,A\,B}{\BaseOpn}{\BaseAlg.\TmInr}\\
    \Alg.\TmInr\,b = \GlueEl{\TmInr\,b}{\BaseOpn}{\BaseAlg.\TmInr\,b}
  \end{mathblock}
\]

The case statement is defined as follows:
\[
  \begin{mathblock}
    \Alg.\TmCase : \Ext{\Alg.\TpSum\,A\,B \to \prn{A\to C} \to \prn{B \to C} \to C}{\BaseOpn}{\BaseAlg.\TmCase}\\
    \Alg.\TmCase\, u\, g\, h =
    \begin{cases}
      \BaseOpn \hookrightarrow \BaseAlg.\TmCase\,u\,g\,h\\
      \Unglue{\BaseOpn}\,u = \ClIntro{\BaseOpn}\TmInl\,{x} \hookrightarrow g\,x\\
      \Unglue{\BaseOpn}\,u = \ClIntro{\BaseOpn}\TmInr\,{x} \hookrightarrow h\,x
    \end{cases}
  \end{mathblock}
\]

The well-definedness of the case split above follows from the computation rules for $\BaseAlg.\TmCase$.

\subsubsection{The sealing modality}

We interpret each $\Alg.\LvlPol{l}$ as $\BaseAlg.\LvlPol{l} =
\CmpAlg.\LvlPol{l} \land \SynAlg.\LvlPol{l}$.
\[
  \begin{mathblock}
    \brk{\TpSeal{l}} : \Ext{\Alg.\Tp\Pos\to\UU}{\BaseOpn}{\BaseAlg.\prn{\Tm\circ \TpSeal{l}}}\\
    \brk{\TpSeal{l}}\,A =
    \GlueFam{\BaseOpn}{
      u : \BaseAlg.\TpSeal{l}A
    }{
      \ClMod{\prn{\BaseOpn\lor\Alg.\LvlPol{l}}}{
        \Compr{
          a:A
        }{
          \OpMod{\BaseOpn}\prn{
            u = \BaseAlg.\TmSeal{l}a
          }
        }
      }
    }
    \\[6pt]
    \Alg.\TpSeal{l} : \Ext{\Alg.\Tp\Pos\to\Alg.\Tp\Pos}{\BaseOpn}{\BaseAlg.\TpSeal{l}}\\
    \Alg.\TpSeal{l}A =
    \GlueEl{
      \brk{\TpSeal{l}}\,A
    }{\BaseOpn}{\BaseAlg.\TpSeal{l}}
  \end{mathblock}
\]

By definition, $\Alg.\TpSeal{l}$ is $\Alg.\LvlPol{l}$-sealed; we define its introduction and elimination form next.
\[
  \begin{mathblock}
    \Alg.\TmSeal{l} : \Ext{A \to \Alg.\TpSeal{l}\,A}{\BaseOpn}{\BaseAlg.\TmSeal{l}}\\
    \Alg.\TmSeal{l}a =
    \GlueEl{
      \ClIntro{\prn{\BaseOpn\lor\Alg.\LvlPol{l}}}{
        a
      }
    }{\BaseOpn}{\BaseAlg.\TmSeal{l}a}
    \\[6pt]
    \Alg.\TmUnseal{l} : \Ext{\brc{B : \Alg.\SealedTp{l}}\to \Alg.\TpSeal{l}A\to \prn{A\to B}\to B}{\BaseOpn}{\BaseAlg.\TmUnseal{l}}\\
    \Alg.\TmUnseal{l}\,\brc{B}\,u\,h =
    \begin{cases}
      \BaseOpn\hookrightarrow \BaseAlg.\TmUnseal{l}\,\brc{B}\,u\,h\\
      \Alg.\LvlPol{l} \hookrightarrow \star\\
      \Unglue{\BaseOpn}{u} = \ClIntro{\prn{\BaseOpn\lor\Alg.\LvlPol{l}}}x \hookrightarrow hx
    \end{cases}
  \end{mathblock}
\]

\subsubsection{Fixed points and declassification}

\begin{lemma}
  If for all $u\lhd\Sub{\Alg.\TpU{X}}u'$ we have $hu\lhd\Sub{\Alg.\TpU{X}}h'u'$, then $\Con{fix}\,h\lhd\Sub{\Alg.\TpU{X}} \SynAlg.\TmFix\,h'$.
\end{lemma}

\begin{proof}
  By Scott induction (\cref{lem:scott-induction}).
\end{proof}

\begin{lemma}
  For any $\Alg.\LvlPol{l}$-sealed value type $A$, if $u\lhd\Sub{\Alg.{\TpSeal{l}\TpU\TpF}A}u'$ then $\CmpAlg.\TmTdcl{l}{u}\lhd\Sub{\Alg.\TpU\TpF{A}}\SynAlg.\TmTdcl{l}{u'}$.
\end{lemma}
\begin{proof}
  We consider the proof of $u\lhd\Sub{\Alg.\prn{\TpSeal{l}\circ\TpU\circ\TpF}\,A}u'$ and proceed by cases:
  \begin{enumerate}

    \item Suppose we have $\Alg.\LvlPol{l} = \top$; this case is trivial,
      because both sides terminate with the unique element of $A$ under
      $\Alg.\LvlPol{l}$.

    \item Suppose we some $v \lhd\Sub{\brk{\Alg.\TpF}\,A} v'$ such that $u =
      \ClIntro{\CmpAlg.\LvlPol{l}}\,{v}$ and $u' = \SynAlg.\TmSeal{l}\,v'$. In this
      case, $\CmpAlg.{\TmTdcl{l}}\,u$ = $\CmpAlg.\LvlPol{l} \parallel v$. To check
      $\prn{\CmpAlg.\LvlPol{l}\parallel v} \lhd\Sub{\brk{\Alg.\TpF}\,A}
      \SynAlg.\TmTdcl{l}\,u$, we may assume that
      $\IsDefd{\prn{\CmpAlg.\LvlPol{l}\parallel{v}}}$ which is the
      $\Sigma$-join $\CmpAlg.\LvlPol{l}\lor\IsDefd{v}$; we
      proceed by cases:
      \begin{enumerate}

        \item Suppose that $\CmpAlg.\LvlPol{l}=\top$.  Because we are proving
          a $\BaseOpn$-sealed proposition,
          \cref{axiom:lvls-under-closed-immersion} allows us to assume that
          $\SynAlg.\LvlPol{l}=\top$. Under this assumption we have
          $\Alg.\LvlPol{l}=\top$, a case that we have already discharged
          above.

        \item If $\IsDefd{v}$, then by assumption we have a
          syntactic point $w' : \OpMod{\SynOpn}{A}$ such that $v\lhd\Sub{A} w'$
          and $v' = \SynAlg.\TmRet\, w'$. Therefore $\SynAlg.\TmTdcl{l}\,u' =
          \SynAlg.\prn{\TmTdcl{l}\circ\TmSeal{l}\circ\TmRet}\,w' =
          \SynAlg.\TmRet\,w'$; on the other side, we have $\CmpAlg.{\TmTdcl{l}}\,u
          = \CmpAlg.{\TmTdcl{l}}\,\ClIntro{\LvlPol{l}}\,v = v$ by
          \cref{lem:stab-action-preserves-ret} using our assumption that
          $\IsDefd{v}$. Therefore our goal follows from our assumption $v\lhd_A w'$.
          \qedhere

      \end{enumerate}

  \end{enumerate}
\end{proof}

\subsubsection{Fundamental Theorem}

\ThmFTLR*

\begin{proof}
  We have provided components for all the generating constants of $\TCat$,
  verifying that all the formal approximation relations assigned to computation
  type connectives are admissible.
\end{proof}

\section{Analytic results: verifying the axioms}\label{apx:analytic}

\subsection{Model of synthetic domain theory}

In this section we give an explicit construction of the SDT topos $\CmpTop$.
Our plan is as follows:

\begin{enumerate}
  \item We first develop a concrete version of domain theory incorporating
    information flow logic based on $\LVL$-indexed dcpos (\cref{sec:dcpos}).

  \item Next we verify in \cref{sec:kadt-model} that the above satisfies the assumptions of a small
    Kleisli model of axiomatic domain theory in the sense of
    \citet[Fiore and Plotkin]{fiore-plotkin:1996}.

  \item Finally in \cref{sec:sdt-model} we adapt the conservative extension result of
    \opcit to obtain a sheaf topos model of synthetic
    domain theory satisfying all the axioms we have imposed.
\end{enumerate}

\subsubsection{Dcpos over the information flow topos}\label{sec:dcpos}

We begin by developing a theory of domains that is \emph{stratified} over
$\LVL$; by this we mean domain theory internal to presheaf category $\Psh{\LVL}
= \brk{\OpCat{\LVL},\SET}$. Such a stratified domain theory is a specialization
of the work of \citet[De Jong and Escard\'o]{dejong-escardo:2021} on
constructive domain theory in homotopy type theory (noting that \opcit made no
use of homotopical axioms that do not hold for a 1-topos). We recall the basic
definitions here, but do not belabor them.
Everything in this section should be read \emph{internally} to $\Psh{\LVL}$;
standard techniques can be used to translate such statements to ordinary
mathematics, as explained by \citet[Mac Lane and
Moerdijk]{maclane-moerdijk:1992} as well as \citet[Awodey
\etal]{awodey-gambino-hazratpour:2021}.

\begin{definition}
  A \DefEmph{directed-complete partial order} or \DefEmph{dcpo} is a partial order
  $A$ that is closed under suprema of directed subsets $S\subseteq A$. A
  continuous function of dcpos is one that preserves these directed suprema.
\end{definition}

\begin{definition}
  A \DefEmph{pointed dcpo} or \DefEmph{dcppo} is a dcpo that possesses a bottom
  element, \ie an element $\bot$ such that $\bot\leq u$ for all $u: A$. A
  \emph{strict} morphism between dcppos is a continuous morphism between the
  underlying dcpos that preserves the bottom element.
\end{definition}

\begin{lemma}
  The forgetful functor from the category of dcppos and strict maps to the
  category of dcpos and continuous maps has a left adjoint $\Lift$
  called the ``lift''.  The underlying set of $\Lift{A}$ is given by the
  \DefEmph{partial element classifier} in $\Psh{\LVL}$, \ie we have
  $\vrt{\Lift{A}} \cong \prn{\phi:\Omega}\times {\OpMod{\phi}\vrt{A}}$.

  \begin{notation}
    Given an element $u:\Lift{A}$, we will write $\IsDefd{u} : \Omega$ for the
    \DefEmph{support} of $u$.  Because the unit map $\Mor|>->|[\eta_A]{A}{\Lift{A}}$
    is a monomorphism, our notation will silently identify elements $u : A$ with
    their images $\prn{\top,\lambda\_.u} : \Lift{A}$.
    We will write $\bot:\Lift{A}$ for the element $\prn{\bot, \lambda\_. \brk{}}$;
  \end{notation}

  We specify that $u\leq\Sub{\Lift{A}}v$ if and only if whenever $u = \eta_Au'$
  and $v = \eta_Av'$ we have $u'\leq\Sub{A}v'$.
\end{lemma}

\begin{notation}
  Given a morphism $\Mor[f]{A}{X}$ of dcpos where $X$ is a dcppo, we will
  write $\Mor[f^\sharp]{\Lift{A}}{X}$ for the unique extension of $f$ as a strict
  map between dcppos.
\end{notation}

We have a \DefEmph{Sierpi\'nski space} $\Sigma = \Lift{\ObjTerm}$ whose carrier is
the subobject classifier $\Omega$ itself. It is not difficult to see that the
projection map $\Mor[\prn{\IsDefd{-}}]{\Lift{A}}{\Sigma}$ is both continuous and
strict.
In addition to the usual suspects of domain theory, we also have \emph{new}
objects that come from the stratification over $\LVL$. In particular, for every
$l\in\LVL$ we have a new element $\Yo[\LVL]{l} : \Sigma$.
Every proposition $\phi:\Sigma$ can be thought of (somewhat degenerately) as a
dcpo; there is exactly one partial order on $\phi$, and as soon as we have
a directed subset of $\phi$ we also have $\phi=\top$, since directed subsets
are stipulated to be non-empty.

\subsubsection{A small Kleisli model of axiomatic domain theory}\label{sec:kadt-model}

We have a fibered category of dcpos over $\Sh{\LvlTop}$; restricting to the
full internal subcategory of $\Sh{\LvlTop}$ spanned by some universe
$\UU\in\Sh{\LvlTop}$ in the sense of
\citet[Hofmann and Streicher]{hofmann-streicher:1997,streicher:2005} and considering the fiber over
the terminal object $\ObjTerm{\Sh{\LvlTop}}$, we obtain a small category
$\dcpo$ of small(er) $\LVL$-indexed dcpos and dcppos.

The (internal) lifting monad restricts to an ordinary lifting monad
$\Mor[\Lift]{\dcpo}{\dcpo}$. We will write $\dcppo$ for the Eilenberg--Moore
category of $\Lift$ and $\pdcpo$ for the Kleisli category of $\Lift$. Here is our glossary:

\begin{enumerate}

  \item The category $\dcpo$ consists of ``predomains'' $A,B$ and continuous
    maps $\Mor{A}{B}$.

  \item The category $\dcppo$ consists of ``domains'' $X,Y$ and \DefEmph{strict}
    continuous maps $\SMor{A}{B}$.

  \item The category $\pdcpo$ consists of predomains $A,B$ and \DefEmph{partial
    maps} $\PrtMor{A}{B}$; equivalently, strict continuous maps
    $\SMor{\Lift{A}}{\Lift{B}}$.

\end{enumerate}

\begin{remark}
  Each of $\dcpo,\dcppo,\pdcpo$ is $\dcpo$-enriched in a canonical way.
\end{remark}

In classical domain theory every $\Lift$-algebra/dcppo/domain is free so we may
identify $\Con{dcppo}$ with $\Con{pdcpo}$; this does \emph{not} hold in the domain theory
indexed over $\LVL$. The reason is that lifting in this setting does much more
than adding a single point; for instance, $\Sigma=\Lift{\ObjTerm}$ contains
strictly more distinct global points than there are elements of $\LVL$.

\paragraph{Preliminaries on embedding-projection pairs}

\begin{definition}
  In an poset-enriched category $\ECat$, we define an \DefEmph{embedding}
  $\EMor{U}{A}$ to be a monomorphism $\Mor|>->|[\epsilon]{U}{A}$ that has a
  right adjoint, called its \DefEmph{projection} $\PMor[\pi]{A}{U}$. We refer to
  the pair $\epsilon\dashv \pi$ as an \DefEmph{embedding-projection pair} or
  \DefEmph{ep-pair}.
\end{definition}

\begin{observation}
  Note that any $\prn{\epsilon,\pi}$ is an embedding-projection pair if and
  only if $\pi\circ\epsilon=\ArrId{U}$ and $\epsilon\circ\pi\leq\ArrId{A}$.
\end{observation}

\begin{definition}
  In a poset-enriched category $\ECat$, an e-initial object is defined to be an
  initial object $\ObjInit\in\ECat$ such that every morphism
  $\Mor{\ObjInit}{E}$ is an embedding. Dually a p-terminal object is defined to
  be a terminal object $\ObjTerm\in\ECat$ such that every morphism
  $\Mor{E}{\ObjTerm}$ is a projection. An object $\ObjZero\in\ECat$ is called
  an \DefEmph{ep-zero object} when it is both e-initial and p-terminal.
\end{definition}

\begin{notation}
  For a poset-enriched category $\ECat$, we will write
  $\EP{\ECat}\subseteq\ECat$ for the wide subcategory spanned by embeddings and
  $\PE{\ECat}\subseteq\ECat$ for the wide subcategory spanned by projections.
\end{notation}

It is not difficult to verify that $\OpCat{\prn{\EP{\ECat}}} = \PE{\ECat}$ and
vice versa~\citep{smyth-plotkin:1982}.

\paragraph{A monadic base}

We verify that $\dcpo$ gives rise to a \emph{monadic base} in the sense of \citet[Fiore and Plotkin]{fiore-plotkin:1996}.

\begin{definition}[\opcit]\label{def:monadic-base}
  A \DefEmph{monadic base} is a
  cartesian closed category $\CCat$ with an initial object, a dominance
  $\Sigma$ whose lift monad is written $\Lift$, and an \DefEmph{inductive fixed point
  object} $\ShLub$ such that the Eilenberg--Moore category
  $\CCat\Sup{\Lift}$ is closed under tensor products and linear homs.
\end{definition}

The category $\dcpo$ is cartesian closed and has an initial object. The
universal family of our dominance is the open inclusion
$\Mor|>->|[\top]{\ObjTerm{\dcpo}}{\Sigma}$.

\begin{lemma}
  The final coalgebra $\ShLub$ for
  the lift monad in $\dcpo$ is an inductive fixed point object, \ie it is the colimit of the
  following chain:
  \[
    \begin{tikzpicture}[diagram]
      \node (0) {$\ObjInit$};
      \node (1) [right = of 0] {$\Lift{\ObjInit}$};
      \node (2) [right = of 1] {$\Lift^2{\ObjInit}$};
      \node (3) [right = of 2] {$\ldots$};
      \draw[>->] (0) to node [above] {$!$} (1);
      \draw[>->] (1) to node [above] {$\Lift{!}$} (2);
      \draw[>->] (2) to node [above] {$\Lift^2{!}$} (3);
    \end{tikzpicture}
  \]
\end{lemma}

\begin{proof}
  The final coalgebra $\ShLub\cong \Lift{\ShLub}$ in $\dcpo$ can be computed as the limit of
  the following $\OpCat{\omega}$-indexed diagram in the category $\PE{\dcpo}$
  of projections~\citep{smyth-plotkin:1982}:
  \[
    \begin{tikzpicture}[diagram]
      \node (0) {$\ObjTerm$};
      \node (1) [right = 2.5cm of 0] {$\Lift{\ObjTerm}$};
      \node (2) [right = 2.5cm of 1] {$\Lift^2{\ObjTerm}$};
      \node (3) [right = 2.5cm of 2] {$\ldots$};
      \draw[<-] (0) to node [above] {$\pi_1 = {!}$} (1);
      \draw[<-] (1) to node [above] {$\pi_2 = \Lift{\pi_1}$} (2);
      \draw[<-] (2) to node [above] {$\pi_3 = \Lift{\pi_2}$} (3);
    \end{tikzpicture}
  \]

  The embeddings $\epsilon_i\dashv \pi_i$ corresponding to each of the
  projections above can be computed recursively, defining $\epsilon_1 = \bot$
  and $\epsilon_{n+1} = \Lift{\epsilon_n}$.
  By the limit-colimit coincidence for diagrams of embedding-projection pairs
  in $\dcpo$, verified in a topos-valid way by \citet[De Jong and Escard\'o]{dejong-escardo:2021}, we
  see that $\ShLub$ is \emph{also} the colimit in $\EP{\dcpo}$ of the
  following diagram of embeddings:
  \[
    \begin{tikzpicture}[diagram]
      \node (0) {$\ObjTerm$};
      \node (1) [right = 2.5cm of 0] {$\Lift{\ObjTerm}$};
      \node (2) [right = 2.5cm of 1] {$\Lift^2{\ObjTerm}$};
      \node (3) [right = 2.5cm of 2] {$\ldots$};
      \draw[embedding] (0) to node [above] {$\epsilon_1 = {\bot}$} (1);
      \draw[embedding] (1) to node [above] {$\epsilon_2 = \Lift{\epsilon_1}$} (2);
      \draw[embedding] (2) to node [above] {$\epsilon_3 = \Lift{\epsilon_2}$} (3);
    \end{tikzpicture}
  \]

  Because $\ObjTerm = \Lift{\ObjInit}$, we may rewrite our diagram like so:
  \[
    \begin{tikzpicture}[diagram]
      \node (0) {$\Lift{\ObjInit}$};
      \node (1) [right = 2.5cm of 0] {$\Lift^2{\ObjInit}$};
      \node (2) [right = 2.5cm of 1] {$\Lift^3{\ObjInit}$};
      \node (3) [right = 2.5cm of 2] {$\ldots$};
      \draw[embedding] (0) to node [above] {$\epsilon_1 = {\bot}$} (1);
      \draw[embedding] (1) to node [above] {$\epsilon_2 = \Lift{\epsilon_1}$} (2);
      \draw[embedding] (2) to node [above] {$\epsilon_3 = \Lift{\epsilon_2}$} (3);
    \end{tikzpicture}
  \]

  The above is not only colimiting in $\EP{\dcpo}$ but also in
  $\dcpo$~\citep{smyth-plotkin:1982}. Moreover it clearly remains colimiting
  when a further map $\Mor|>->|{\ObjInit}{\Lift{\ObjInit}}$ is adjoined to the
  left, so we are done.
\end{proof}

Finally we verify the closure
of the Eilenberg--Moore category $\pdcpo$ under tensor products and linear
homs:
\begin{enumerate}

  \item Given two domains $X,Y$ we define the tensor/smash product $X\otimes Y$
    to be the quotient of $X\times Y$ by the relation that identifies
    $\prn{\bot,y} \sim \prn{x,\bot}\sim \prn{\bot,\bot}$.

  \item Given two domains $X,Y$, the linear hom $X\multimap Y$ is the subobject
    of the exponential predomain $X\to Y$ spanned by \emph{strict} continuous
    maps.

\end{enumerate}

\begin{corollary}
  The category $\dcpo$ with its dominance $\Sigma$ and lifting monad $\Lift$
  forms a monadic base in the sense of \citet[Fiore and Plotkin]{fiore-plotkin:1996}.
\end{corollary}

\paragraph{Algebraic compactness and the Kleisli model}

The following result is easily adapted from \citet[Fiore's dissertation]{fiore:1994}.

\begin{fact}\label{fact:compactness-from-zero-and-bilimits}
  Let $\ECat$ be $\dcpo$-enriched; if $\ECat$ has an ep-zero object and
  $\EP{\ECat}$ is closed under colimits of $\omega$-chains, then $\ECat$ is
  $\dcpo$--algebraically compact.
\end{fact}

By $\dcpo$--algebraic compactness, we mean that every $\dcpo$-enriched
endofunctor $\Mor[F]{\ECat}{\ECat}$ has a \emph{free} algebra, \ie an object that
simultaneously carries $F$'s initial algebra and final coalgebra.

\begin{definition}[{\citet[Fiore and Plotkin]{fiore-plotkin:1996}}]
  Letting $\prn{\CCat,\Sigma,\Lift}$ be a monadic base in the sense of
  \cref{def:monadic-base}, the Kleisli category $\CCat\Sub{\Lift}$ is said to
  be a \DefEmph{Kleisli model of axiomatic domain theory} when it is
  $\CCat$-algebraically compact.
\end{definition}

We leave to the reader the following easily verified fact.

\begin{fact}
  The Kleisli category $\pdcpo = \prn{\dcpo}\Sub{\Lift}$ has an ep-zero object.
\end{fact}

\begin{lemma}
  The category of partial embeddings $\EP{\pdcpo}$ is closed under colimits of
  $\omega$-chains.
\end{lemma}

\begin{proof}
  By Theorem~5.3.14 of \citet[Fiore's dissertation]{fiore:1994} and the fact
  that the lift monad is $\dcpo$-enriched and therefore preserves ep-chains, it
  suffices to recall that $\EP{\dcpo}$ is closed under colimits of
  $\omega$-chains~\citep{dejong-escardo:2021}.
\end{proof}

\begin{remark}
  For intuition,
  \citet[Sterling]{sterling:2022:bilimits} has also verified an explicit (Grothendieck) topos--valid computation
  of the colimit of an arbitrary diagram $\Mor[A_\bullet]{\ICat}{\EP{\pdcpo}}$
  where $\ICat$ is a small filtered poset, generalizing the claim of
  \citet[Jones and Plotkin]{jones:1990,jones-plotkin:1989}:
  \[
    A_\infty \coloneqq
    \Compr{
      \sigma : \Prod{i\in\ICat}\Lift{A_i}
    }{
      \prn{\exists i\in\ICat. \IsDefd{\sigma_i}}
      \land
      \forall i\leq j\in\ICat.
      \pi\Sub{i\leq j}\sigma_j = \sigma_i
    }
    \qedhere
  \]
\end{remark}

\begin{corollary}
  The category of dcpos and partial maps $\pdcpo$ is $\dcpo$-algebraically
  compact, and hence a Kleisli model of axiomatic domain theory.
\end{corollary}

\subsubsection{Dominances and orthogonality}\label{sec:dominances-and-orthogonality}

Let $\ECat$ be a locally cartesian closed category equipped with a dominance
$\Sigma$.  Let $\mathcal{M}$ be a class of monomorphisms in $\ECat$; we define $\mathcal{M}_\times$ be the smallest class of
monomorphisms in $\ECat$ stable under products and containing
$\mathcal{M}$, and we define $\mathcal{M}_\Sigma$ be the smallest class of
monomorphisms in $\ECat$ stable under pullback along
$\Sigma$-monomorphisms and containing $\mathcal{M}_\Sigma$.

\begin{notation}
  We will write $m\perp f$ to mean that $f$ is externally orthogonal to $m$;
  we will write $\mathcal{M}\perp f$ to mean that $f$ is that $m\perp f$ for
  every $m\in\mathcal{M}$.
\end{notation}

\begin{fact}
  An object $E\in\ECat$ is internally orthogonal to $\mathcal{M}$ if and only if $\mathcal{M}_\times \perp E$.
\end{fact}

\begin{lemma}\label{fact:orth-to-orth-lift}
  If $\mathcal{M}_\times\perp \Sigma$ and $\mathcal{M}_\Sigma\perp E$, then
  $\mathcal{M}_\times\perp \Lift{E}$.
\end{lemma}

\begin{proof}
  Fix $\Mor|>->|{I}{J}\in\mathcal{M}_\times$ and a lifting problem of the following form:
  \[
    \begin{tikzpicture}[diagram]
      \node (I) {$I$};
      \node (J) [below = of I] {$J$};
      \node (L/E) [right = of I] {$\Lift{E}$};
      \draw[>->] (I) to (J);
      \draw[->] (I) to node [above] {$a$} (L/E);
      \draw[exists,->] (J) to node [sloped,below] {$?$} (L/E);
    \end{tikzpicture}
  \]

  The map $\Mor[a]{I}{\Lift{E}}$ corresponds to a unique total map
  $\Mor[\tilde{a}]{I_a}{E}$ defined on a Scott-open subset $I_a\subseteq I$;
  because $\mathcal{M}_\times\perp\Sigma$, we may compute the support of the
  desired map $\Mor|exists,->|{J}{\Lift{E}}$ by solving another lifting problem:
  
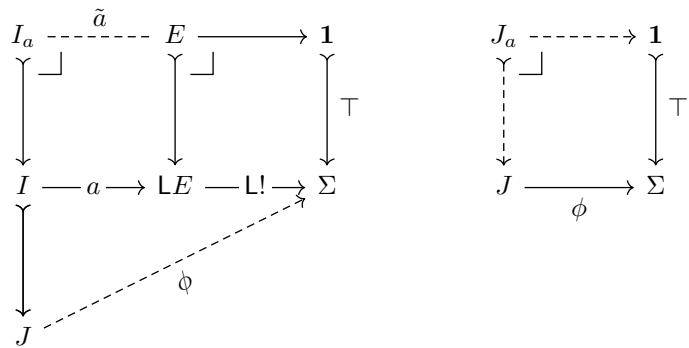
\[
    \begin{tikzpicture}[diagram,baseline=(l/sw.base)]
      \SpliceDiagramSquare<l/>{
        nw = I_a,
        ne = E,
        sw = I,
        se = \Lift{E},
        north = \tilde{a},
        south = a,
        west/style = >->,
        east/style = >->,
        north/style = exists,
        nw/style = pullback,
        ne/style = pullback,
        south/node/style = upright desc,
      }
      \SpliceDiagramSquare<r/>{
        glue = west,
        glue target = l/,
        ne = \ObjTerm,
        se = \Sigma,
        east = \top,
        south = \Lift{!},
        east/style = >->,
        south/node/style = upright desc,
      }
      \node (J) [below = of l/sw] {$J$};
      \draw[>->] (l/sw) to (J);
      \draw[>->] (l/sw) to (J);
      \draw[exists,->] (J) to node [sloped,below] {$\phi$} (r/se);
    \end{tikzpicture}
    \qquad\qquad
    \DiagramSquare{
      nw = J_a,
      sw = J,
      ne = \ObjTerm,
      se = \Sigma,
      east = \top,
      east/style = >->,
      west/style = {exists,>->},
      north/style = {exists,->},
      nw/style = pullback,
      south = \phi,
    }
  \]

  Using the pullback lemma we deduce that the left-hand square below is cartesian:
  
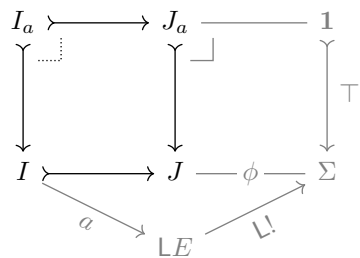
\[
    \begin{tikzpicture}[diagram]
      \SpliceDiagramSquare<l/>{
        nw = I_a,
        ne = J_a,
        sw = I,
        se = J,
        north/style = >->,
        south/style = >->,
        west/style = >->,
        east/style = >->,
        nw/style = dotted pullback,
        ne/style = muted pullback,
      }
      \SpliceDiagramSquare<r/>{
        glue = west,
        glue target = l/,
        east/style = {gray,>->},
        north/style = gray,
        south/style = gray,
        south/node/style = upright desc,
        ne/style = gray,
        se/style = gray,
        east = \top,
        ne = \ObjTerm,
        se = \Sigma,
        south = \phi,
      }
      \node (L/E) [gray,below = 1cm of l/se] {$\Lift{E}$};
      \draw[->,gray] (l/sw) to node [sloped,below] {$a$} (L/E);
      \draw[->,gray] (L/E) to node [sloped,below] {$\Lift{!}$} (r/se);
    \end{tikzpicture}
  \]

  Therefore $\Mor|>->|{I_a}{J_a}\in\mathcal{M}_\Sigma$ and so we may solve the
  following lifting problem in $E$, and use the universal property of the
  partial map classifier to obtain the desired lift:
  
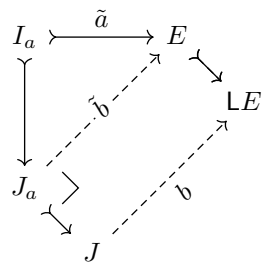
\[
    \begin{tikzpicture}[diagram,baseline=(L/E.base)]
      \node (I/a) {$I_a$};
      \node (J/a) [pullback 45,below = of I/a] {$J_a$};
      \node (E) [right = of I/a] {$E$};
      \node (J) [below right = 1.25cm of J/a] {$J$};
      \node (L/E) [below right = 1.25cm of E] {$\Lift{E}$};
      \draw[>->] (I/a) to (J/a);
      \draw[>->] (I/a) to node [above] {$\tilde{a}$} (E);
      \draw[>->] (J/a) to (J);
      \draw[>->] (E) to (L/E);
      \draw[exists,>->] (E) to(L/E);
      \draw[exists,->] (J/a) to node [desc] {$\tilde{b}$} (E);
      \draw[exists,->] (J) to node [sloped,below] {$b$} (L/E);
    \end{tikzpicture}
    \qedhere
  \]

\end{proof}

The following lemma appears as Proposition~3.2(2) in
\citet[Fiore and Plotkin]{fiore-plotkin:1996}, but is left unproved in \opcit.

\begin{lemma}\label{lem:orth-incl-creates-limits}
  The inclusion of any orthogonal subcategory creates limits.
\end{lemma}

\begin{proof}
  Let $\mathcal{N}$ be a class of monomorphisms and let $\ECat\Sub{\mathcal{N}\perp}$ be
  the full subcategory of $\ECat$ spanned by objects $E$ such that
  $\mathcal{N}\perp{E}$.
  Let $\Mor[E_\bullet]{\ICat}{\ECat\Sub{\mathcal{N}\perp}}$ be a diagram that
  has a limit in $\ECat$; then we will argue that
  $\mathcal{N}\perp\Lim{\ICat}{E_\bullet}$.  Fixing
  $\Mor|>->|{U}{V}\in\mathcal{N}$, consider any lifting problem of the
  following kind in $\ECat$:
  \[
    \begin{tikzpicture}[diagram]
      \node (U) {$U$};
      \node (V) [below = of U] {$V$};
      \node (Lim) [right = of U] {$\Lim{\ICat}{E_\bullet}$};
      \draw[>->] (U) to (V);
      \draw[->] (U) to (Lim);
      \draw[exists,->] (V) to node [sloped,below] {$?$} (Lim);
    \end{tikzpicture}
  \]

  By the universal property of the limit, a solution to the lifting problem
  above in $\ECat$ is uniquely determined by a solution to the following
  lifting problem in $\brk{\ICat,\ECat}$.
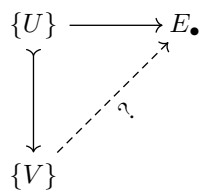
  \[
    \begin{tikzpicture}[diagram]
      \node (U) {$\brc{U}$};
      \node (V) [below = of U] {$\brc{V}$};
      \node (EA) [right = of U] {$E_\bullet$};
      \draw[>->] (U) to (V);
      \draw[->] (U) to (E);
      \draw[exists,->] (V) to node [sloped,below] {$?$} (E);
    \end{tikzpicture}
  \]

  Such a lifting problem can be solved pointwise, recalling that each $E_i$ is orthogonal to $\mathcal{N}$:
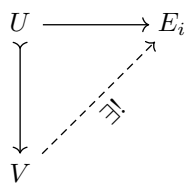
  \[
    \begin{tikzpicture}[diagram]
      \node (U) {$U$};
      \node (V) [below = of U] {$V$};
      \node (EA) [right = of U] {$E_i$};
      \draw[>->] (U) to (V);
      \draw[->] (U) to (E);
      \draw[exists,->] (V) to node [sloped,below] {$\exists!$} (E);
    \end{tikzpicture}
  \]

  Naturality of the induced cone $\Mor{\brc{V}}{E_\bullet}$ follows from the
  uniqueness of lifts.
\end{proof}

\begin{lemma}\label{fact:orth-lift-to-orth}
  Let $\mathcal{N} \in \brc{\mathcal{M}_\times,\mathcal{M}_\Sigma,\mathcal{M}}$.
  If $\mathcal{N}\perp\Sigma$ and $\mathcal{N}\perp\Lift{E}$ then $\mathcal{N}\perp{E}$;
\end{lemma}

\begin{proof}
  By \cref{lem:orth-incl-creates-limits}, since $E$ is the following pullback:
  \[
    \DiagramSquare{
      nw/style = pullback,
      nw = E,
      ne = \ObjTerm{\ECat},
      se = \Sigma,
      sw = \Lift{E},
      east/style = >->,
      west/style = >->,
    }
    \qedhere
  \]
\end{proof}

\begin{lemma}\label{lem:orth-lift-to-stab-orth-lift}
  If $\mathcal{M}_\times\perp\Lift{E}$, then also $\mathcal{M}_\Sigma\perp\Lift{E}$.
\end{lemma}

\begin{proof}
  Fix $\Mor|>->|{I}{J}\in\mathcal{M}_\times$ and let $\Mor|>->|{V}{J}$ be a
  $\Sigma$-monomorphism; consider a lifting problem of the following form:
  
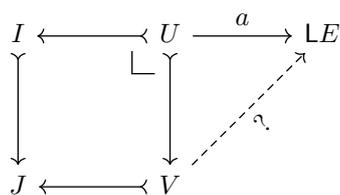
\[
    \begin{tikzpicture}[diagram]
      \node (U) [ne pullback] {$U$};
      \node (V) [below = of U] {$V$};
      \node (L/E) [right = of U] {$\Lift{E}$};
      \node (I) [left = of U] {$I$};
      \node (J) [left = of V] {$J$};
      \draw[>->] (U) to (V);
      \draw[>->] (I) to (J);
      \draw[>->] (V) to (J);
      \draw[>->] (U) to (I);
      \draw[->] (U) to node [above] {$a$} (L/E);
      \draw[exists,->] (V) to node [sloped,below] {$?$} (L/E);
    \end{tikzpicture}
  \]

  We may factor $\Mor[a]{U}{\Lift{E}}$ through $I$ as follows:
  
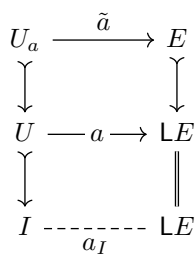
\[
    \begin{tikzpicture}[diagram]
      \node (U/a) {$U_a$};
      \node (E) [right = of U/a] {$E$};
      \node (U) [below = 1.25cm of U/a] {$U$};
      \node (L/E) [below = 1.25cm of E] {$\Lift{E}$};
      \node (L/E') [below = 1.25cm of L/E] {$\Lift{E}$};
      \node (I) [below = 1.25cm of U] {$I$};
      \draw[>->] (E) to (L/E);
      \draw[>->] (U/a) to (U);
      \draw[->] (U) to node [upright desc] {$a$} (L/E);
      \draw[->] (U/a) to node [above] {$\tilde{a}$} (E);
      \draw[>->] (U) to (I);
      \draw[double] (L/E) to (L/E');
      \draw[exists] (I) to node [below] {$a_I$} (L/E');
    \end{tikzpicture}
  \]

  Therefore we may solve the following lifting problem:
  
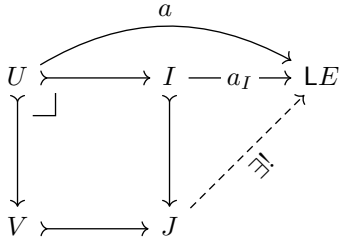
\[
    \begin{tikzpicture}[diagram,baseline=(J.base)]
      \node (I) {$I$};
      \node (J) [below = of I] {$J$};
      \node (L/E) [right = of I] {$\Lift{E}$};
      \node (U) [pullback,left = of I] {$U$};
      \node (V) [left = of J] {$V$};
      \draw[>->] (I) to (J);
      \draw[>->] (U) to (V);
      \draw[>->] (U) to (I);
      \draw[>->] (V) to (J);
      \draw[->] (I) to node [upright desc] {$a_I$} (L/E);
      \draw[->,exists] (J) to node [sloped,below] {$\exists!$} (L/E);
      \draw[->,bend left = 30] (U) to node [above] {$a$} (L/E);
    \end{tikzpicture}
    \qedhere
  \]
\end{proof}

\begin{lemma}\label{lem:lift-orth-reflective}
  Let $\ECat$ be locally presentable.
  If $\mathcal{M}_\times\perp\Sigma$, then the subcategory of $\ECat$ spanned
  by $E$ such that $\mathcal{M}_\times\perp\Lift{E}$ is reflective.
\end{lemma}

\begin{proof}
  In a locally presentable category, any orthogonal subcategory is reflective.
  We therefore argue that $\mathcal{M}_\times\perp\Lift{E}$ if and only if
  $\mathcal{M}_\Sigma\perp{E}$.

  \begin{enumerate}

    \item Suppose that $\mathcal{M}_\Sigma\perp{E}$; then
      \cref{fact:orth-to-orth-lift} applies as we have also assumed
      $\mathcal{M}_\times\perp\Sigma$.

    \item Suppose that $\mathcal{M}_\times\perp\Lift{E}$; by
      \cref{fact:orth-lift-to-orth} it suffices to check that
      $\mathcal{M}_\Sigma\perp\Sigma$ and $\mathcal{M}_\Sigma\perp\Lift{E}$
      which both follow from our assumptions via
      \cref{lem:orth-lift-to-stab-orth-lift}.

  \end{enumerate}
\end{proof}

\NewDocumentCommand\FAM{m}{\mathbf{Fam}\Sub{\SET}\prn{#1}}
\paragraph{Definition of cartesian coverages and sheaves}

\begin{definition}
  Let $\CCat$ be a category and let $u\in \CCat$; a \emph{sink} on $u$ is
  defined to be a collection of morphisms into $u$, \ie an element of
  $\FAM{\Sl{\CCat}{u}}$.
\end{definition}

When $\CCat$ has pullbacks, we may consider the base change of a sink
$\underline{u} = \prn{\Mor{u_i}{u}\mid i\in I}$ along a map $\Mor{v}{u}$, namely the family
$v^*\underline{u} = \prn{\Mor{v\times\Sub{u}u_i}{v}\mid i\in I}$.

\begin{definition}
  Let $\CCat$ be a category with pullbacks. Then a \emph{cartesian coverage} on
  $\CCat$ is given by an assignment to each object $u\in\CCat$ a collection
  $K\prn{u}\subseteq\FAM{\Sl{\CCat}{u}}$ of sinks on $u$ that is stable under
  base change: if $\underline{u}\in K\prn{u}$ then for any $\Mor{v}{u}$, the
  sink $v^*\underline{u}$ lies in $K\prn{v}$.
\end{definition}

We refer to an element of $K\prn{u}$ as a \emph{covering sink on $u$}.

\begin{definition}
  Let $\CCat$ be category with pullbacks and let $\underline{u}$ be a sink on
  $u\in \CCat$. For a presheaf $F\in \Psh{\CCat}$, a \emph{matching family} for
  $\underline{u}$ is given by an assignment of elements $x_i\in F\prn{u_i}$ to each $\Mor{u_i}{u}$ in $\underline{u}$, such that for all $u_i,u_j\in \underline{u}$ we have $u\Sub{ij}^*x_i = u\Sub{ij}^*x_j\in
  F\prn{u\Sub{ij}}$ where $u\Sub{ij}$ is the fiber product $u_i\times\Sub{u} u_j\in\Sl{\CCat}{u}$.
\end{definition}

\begin{definition}
  Let $\CCat$ be a category with pullbacks and let $K$ be a coverage on
  $\CCat$. A presheaf $F\in \Psh{\CCat}$ is defined to be a \emph{sheaf}
  relative to $K$ when for any covering sink $\underline{u}\in K\prn{u}$, there
  is a bijection between matching families for $\underline{u}$ in $F$ and
  elements of $F\prn{u}$.
\end{definition}

\subsubsection{A sheaf model of synthetic domain theory}\label{sec:sdt-model}

Taking the Kleisli model of axiomatic domain theory from the previous section,
we may now adapt the conservative extension theorem of
\citet[Fiore and Plotkin]{fiore-plotkin:1996} to embed it into a sheaf model of synthetic domain
theory.

\begin{notation}
  In this section, we will write $\CmpCat$ for $\dcpo$.
\end{notation}

\paragraph{Descent properties of dpcos}

We begin by showing that two classes of colimit in $\CmpCat$ enjoy a useful
descent property, which will enable us to embed $\CmpCat$ into a topos such
that the resulting Yoneda embedding preserves these colimits.
Monomorphisms of dcpos are not generally well-behaved. But there is a dominion
of monomorphisms in $\CmpCat$ that we shall call \emph{Scott-open immersions}
that has a number of very useful properties.

\begin{definition}
  A \DefEmph{Scott-open immersion} $\Mor|open immersion|{U}{A}$ of dcpos is any
  map that arises by pullback from the dominance $\Mor|open
  immersion|[\top]{\ObjTerm{\CmpCat}}{\Sigma}$. (In the category of dcpos, we shall use this special arrow to denote open immersions.)
\end{definition}

Every Scott-open immersion $\Mor|open immersion|[j]{U}{A}$ factors uniquely
through an isomorphism $U\cong \Con{Im}\,j$ and a subdcpo inclusion
$\Con{Im}\,j\subseteq A$ such that the order on $\Con{Im}\,j$ is the
restriction of the order on $A$.  We will refer to subdcpos that arise from
such a factorization as \DefEmph{Scott-open subdcpos}; of course, every
Scott-open subdcpo inclusion is trivially a Scott-open immersion. Therefore we
need not pay much attention to the difference between Scott-open immersions and
Scott-open subdcpos.

\begin{lemma}
  The Scott-open subdcpos are closed under finite unions.
\end{lemma}

\begin{proof}
  We mean that given a finite family of subobjects $\prn{U_i\in\SUB{A} \mid
  i\in I}$ that are each Scott-open, the union $\Disj{i\in I}{U_i}$
  exists and moreover is Scott-open. This is not difficult to see, as we may
  compute $\Disj{i}{U_i}$ using the operation
  $\Mor[\lor]{\Sigma\Sup{I}}{\Sigma}$. Let $\Mor[\phi_i]{A}{\Sigma}$
  be the characteristic maps of each $U_i$; then we have
  $\Mor[\Disj{i}{\phi_i}]{A}{\Sigma}$ which will serve as the characteristic
  map for $\Mor|open immersion|{\Disj{i}{U_i}}{A}$.
\end{proof}

\begin{lemma}\label{lem:stable-unions}
  Finite unions of Scott-open subobjects in $\CmpCat$ are \DefEmph{stable} under pullback.
\end{lemma}

\begin{proof}
  We consider the partition of a dcpo $X\in \CmpCat$ as the union of
  a finite family of Scott-open subdcpos $\Mor|open immersion|{U_i}{X}$ indexed in some finite set $I$. We must check the union cocone is
  preserved by pulling back along some $\Mor[k]{Y}{X}$, \ie we need to reconstruct
  $Y$ as the union of the Scott-open subdcpos $k^*U_i$. That $X = \Disj{i}{U_i}$ means that the following square is cartesian:
  \[
    \begin{tikzpicture}[diagram]
      \SpliceDiagramSquare<l/>{
        west/style = double,
        east/style = open immersion,
        nw/style = pullback,
        ne/style = pullback,
        nw = X,
        sw = X,
        ne = \cdots,
        se = \Sigma\Sup{I},
        south = \phi_\bullet,
      }
      \SpliceDiagramSquare<r/>{
        glue = west,
        glue target = l/,
        east/style = open immersion,
        east = \top,
        ne = \ObjTerm{\CmpCat},
        se = \Sigma,
        south = \prn{\lor},
      }
    \end{tikzpicture}
  \]

  Pulling back further along $\Mor[k]{Y}{X}$, we have:
  \[
    \begin{tikzpicture}[diagram]
      \SpliceDiagramSquare<ll/>{
        nw = Y,
        sw = Y,
        ne = X,
        se = X,
        south = k,
        west/style = double,
        east/style = double,
        nw/style = pullback,
        ne/style = pullback,
      }
      \SpliceDiagramSquare<l/>{
        glue = west,
        glue target = ll/,
        ne = \cdots,
        se = \Sigma\Sup{I},
        east/style = open immersion,
        south = \phi_\bullet,
        ne/style = pullback,
      }
      \SpliceDiagramSquare<r/>{
        glue = west,
        glue target = l/,
        east/style = open immersion,
        east = \top,
        ne = \ObjTerm{\CmpCat},
        se = \Sigma,
        south = \prn{\lor},
      }
    \end{tikzpicture}
  \]

  But the characteristic map $k;\phi_\bullet;\lor$ is equal to
  $\prn{\lambda i. k;\phi_i};\lor$. Therefore $Y$ is the union of $k^*U_i$
  because $k;\phi_i$ is the characteristic map for $k^*U_i$.
\end{proof}

\begin{lemma}\label{lem:dir-set-isect}
  Let $D$ be a directed subset of a dcpo $A\in\CmpCat$ and let $U\subseteq A$
  be a Scott-open subdcpo of $A$ such that $U\cap D$ is inhabited; then $U\cap D$
  is also directed.
\end{lemma}

\begin{proof}
  Fixing $x,y\in U\cap D$ we must find some $z\in U\cap D$ such that $x,y\leq
  z$. Because $D$ is directed there does exist $z\in D$ such that $x,y\leq z$.
  As $U$ is a $\Sigma$-subset, there exists a unique Scott continuous
  characteristic map $\Mor[\phi_U]{A}{\Sigma}$ encoding $U$. A continuous map is
  also monotone, so from $x\leq z$ we conclude $\phi_Ux\leq\phi_Uz$. Therefore
  $z\in U$ follows from $x\in U$.
\end{proof}

\begin{lemma}\label{lem:effective-unions}
  Finite unions of Scott-open subobjects in $\CmpCat$ are \DefEmph{effective} in the sense of
  \citet[Barr]{barr:1988} and \citet[Garner and Lack]{garner-lack:2012}.
\end{lemma}

\begin{proof}
  Let $\Mor|open immersion|{U_i}{X}$ be a finite family of Scott-open
  immersions in $\CmpCat$ indexed in $i\in I$, such that each $U_i$ is classified by $\Mor[\phi_i]{X}{\Sigma}$ respectively.
  We recall that any Scott-open immersion is an order-embedding in the sense
  each $U_i$ has the order induced by restricting that of $X$ to points
  satisfying $\phi_i$. We recall that $\vrt{\Sigma}$ is the subobject
  classifier of $\Sh{\LvlTop}$ and $\Mor[\lor]{\Sigma\Sup{I}}{\Sigma}$
  is tracked in $\Sh{\LvlTop}$ by the actual topos disjunction. Therefore the
  underlying sheaf of points of the union $\Disj{i}{U_i}$ can be computed in
  $\Sh{\LvlTop}$ as the \emph{wide} (effective) pushout of all the
  projections $\Mor|>->|[p_j]{\vrt{\Conj{i}{U_i}} =
  \Conj{i}{\vrt{U_i}}}{\vrt{U_j}}$, as any topos has effective unions of small
  arity.  The effectivity property for wide pushouts of monomorphisms in a
  topos means that $\vrt{\Conj{i}{U_i}}$ is the wide pullback of all the
  $\Mor|>->|{\vrt{U_j}}{\vrt{\Disj{i}{U_i}}}$.

  Then $\Disj{i}{U_i}$ is the Scott-open subdcpo of $X$ spanned by
  $\vrt{\Disj{i}{U_i}}$ with order induced from $X$. We note that
  $\Conj{i}{U_i}$ remains the wide pullback in $\CmpCat$ of all the $\Mor|open
  immersion|{U_j}{\Disj{i}{U_i}}$ because pullbacks of dcpos are computed as in
  $\Sh{\LvlTop}$.  It remains to show that the $\Disj{i}{U_i}$ is \emph{also}
  the wide pushout in $\CmpCat$ of all the $\Mor|open
  immersion|{\Conj{i}{U_i}}{U_j}$. To this end, fix any cocone
  $\Mor[h_\bullet]{U_\bullet}{\brc{Y}}$ in which $Y$ is an arbitrary object of $\CmpCat$;
  at the level of points, there is a unique function
  $\Mor[h_\infty]{\vrt{\Disj{i}{U_i}}}{Y}$ making the following diagram
  commute in $\brk{I,\Sh{\LvlTop}}$:
  \[
    \begin{tikzpicture}[diagram]
      \node (0) {$U_\bullet$};
      \node (1) [right = of 0] {$\brc{\vrt{\Disj{i}{U_i}}}$};
      \node (2) [below = of 1] {$\brc{\vrt{Y}}$};
      \draw[->] (0) to (1);
      \draw[exists,->] (1) to node [right] {$\exists! h_\infty$} (2);
      \draw[->] (0) to node [sloped,below] {$\vrt{h_\bullet}$} (2);
    \end{tikzpicture}
  \]

  Therefore it remains to show that
  $\Mor[h_\infty]{\vrt{\Disj{i}{U_i}}}{\vrt{Y}}$ is Scott continuous and hence
  tracks a morphism $\Mor{\Disj{i}{U_i}}{Y}$ in $\CmpCat$. We fix a directed
  subset $D\subseteq \Disj{i}{U_i}$ to check that $h_\infty\prn{\DLub{}D} =
  \DLub{x\in D}h_\infty\,x$. Note that that $\DLub{}{D}$ must lie in $U_k$ for
  some $k$; observe that $\DLub{}{D}$ is also the least upper bound of the
  subset $D\cap U_k\subseteq U_k$ which remains directed by
  \cref{lem:dir-set-isect}. Because $h_k$ is assumed Scott continuous, we have
  $h_\infty\prn{\DLub{}{D}} = h_k\prn{\DLub{}{D}} = \DLub{x\in D\cap
  U_k}{h_kx}$; but this is equal to $\DLub{x\in D}h_\infty\,x$ because any
  element of $D$ is bounded by an element of $U$, namely $\DLub{}{D}$ itself.
\end{proof}

\begin{lemma}\label{lem:coprod-inj-scott-open}
  Coproduct injections in $\CmpCat$ are Scott-open immersions.
\end{lemma}

\begin{proof}
  Let $I$ be a set and let $A_i$ be a family of dcpos indexed in $i\in I$. For
  $k\in I$ we must verify that the injection $\Mor|>->|{A_k}{\Coprod{i}{A_i}}$
  is in fact a Scott-open immersion. We do so by constructing the unique characteristic map $\Mor[\phi_k]{\Coprod{i}{A_i}}{\Sigma}$ making the following square cartesian:
  \[
    \DiagramSquare{
      ne = \ObjTerm{\CmpCat},
      se = \Sigma,
      sw = \Coprod{i}{A_i},
      nw = A_k,
      nw/style = pullback,
      south/style = {->,exists},
      east/style = open immersion,
      west/style = open immersion,
    }
  \]

  Of course, a map $\Mor{\Coprod{i}{A_i}}{\Sigma}$ is the same as for each
  $i\in I$, a map $\Mor{A_i}{\Sigma}$. When $i=k$ we return $\top\in\Sigma$ and
  otherwise we return $\bot\in\Sigma$.
\end{proof}

\begin{corollary}\label{cor:coproducts-are-unions}
  Any coproduct in $\CmpCat$ is the union of a family of Scott-open immersions.
\end{corollary}

\paragraph{A coverage on the category of dcpos}

We will impose a coverage on $\CmpCat$ that will ensure an embedding into a
Grothendieck topos that preserves finite unions of Scott-open subobjects (and
thus finite coproducts). For any $A\in \CmpCat$ we let $K\prn{A}$ be the set of
sinks given by finitely indexed families Scott-open immersions $\prn{\Mor|open
immersion|{U_i}{A}\mid i\in I}$ such that $A \cong \Disj{i\in I}{U_i}$.

\begin{lemma}
  The assignment $K$ of sets of sinks to objects of $\CmpCat$ is a cartesian coverage.
\end{lemma}

\begin{proof}
  By \cref{lem:stable-unions}.
\end{proof}

\begin{lemma}
  The cartesian coverage $K$ is subcanonical.
\end{lemma}

\begin{proof}
  Let $A\in\CmpCat$ be a dcpo; we must check that $\Yo[\CmpCat]{A}$ is a sheaf
  for $K$. This follows from the fact that unions of Scott-open subobjects are
  effective (\cref{lem:effective-unions}).  Indeed, considering only the binary case for simplicity, we must check
  the collection of pairs of morphisms $\Mor[a_U]{U}{A}$ and $\Mor[a_V]{V}{A}$
  that agree on the intersection $U\land V$ is bijective with the collection of
  morphisms $\Mor{U\lor V}{A}$; this follows immediately by effectivity which guarantees that the following square is both cartesian and cocartesian:
  \[
    \DiagramSquare{
      nw/style = pullback,
      se/style = pushout,
      west/style = open immersion,
      east/style = open immersion,
      north/style = open immersion,
      south/style = open immersion,
      nw = U\land V,
      sw = V,
      ne = U,
      se = U\lor V
    }
    \qedhere
  \]
\end{proof}

\begin{construction}\label{def:dcpo-localization}
  We will write $\EmbMor{\CmpTop}{\PrTop{\CmpCat}}$ for the subtopos determined
  by localizing $\Psh{\CCat}$ at the coverage $K$.
\end{construction}

Because the localization is subcanonical, the Yoneda
embedding $\EmbMor[\Yo[\CmpCat]]{\CmpCat}{\Psh{\CmpCat}}$ factors through a fully
faithful functor $\EmbMor[\TildeYo]{\CmpCat}{\Sh{\CmpTop}}$:
\[
  \begin{tikzpicture}[diagram]
    \node (D) {$\CmpCat$};
    \node (Psh) [right = of D] {$\Psh{\CmpCat}$};
    \node (Sh) [below = of D] {$\Sh{\CmpTop}$};
    \draw[embedding] (D) to node [above] {$\Yo[\CmpCat]$} (Psh);
    \draw[embedding] (D) to node [left] {$\TildeYo$} (Sh);
    \draw[embedding] (Sh) to (Psh);
  \end{tikzpicture}
\]

\begin{fact}
  The subcategory $\Sh{\CmpTop}\subseteq\Psh{\CmpCat}$ is spanned precisely by
  the presheaves that send finite unions of Scott-open subobjects in $\CmpCat$ to
  pullbacks in $\SET$.
\end{fact}

\paragraph{Conservative extension and basic axioms of SDT}

The conservative extension result (Theorem~3.4) of \citet[Fiore and Plotkin]{fiore-plotkin:1996}
can be adapted to our setting to construct a model of synthetic domain theory
in $\Sh{\CmpTop}$ satisfying
\cref{axiom:sigma-finite-joins,axiom:sigma-predomain,axiom:omega-inductive}.
extending the model of axiomatic domain theory constructed in \cref{sec:kadt-model} (we verify
the remaining \cref{axiom:lvls,axiom:finset-transparent-predomain} in \cref{sec:sdt-info-flow}). The
difference between our setting and \opcit is that the latter localize with the 0-ary
extensive topology, whereas we do so with a more sophisticated topology (\cref{def:dcpo-localization}).
Writing $\PREDOM\subseteq\Sh{\CmpTop}$ for the full subcategory spanned by
predomains, the Yoneda embedding $\EmbMor{\CmpCat}{\Sh{\CmpTop}}$ factors through $\PREDOM$ and preserves limits, exponentials, the final lift-coalgebra $\ShLub$, countable coproducts, and the
dominance $\Sigma$; moreover, the lifting monad is extended by each embedding in
the sense that the following squares commute up to isomorphism:
\[
  \begin{tikzpicture}[diagram]
    \SpliceDiagramSquare<0/>{
      north/style = embedding,
      south/style = embedding,
      east/node/style = upright desc,
      nw = \CmpCat,
      ne = \PREDOM,
      sw = \CmpCat,
      se = \PREDOM,
      west = \Lift,
      east = \Lift,
    }
    \SpliceDiagramSquare<1/>{
      glue = west,
      glue target = 0/,
      north/style = embedding,
      south/style = embedding,
      ne = \Sh{\CmpTop},
      se = \Sh{\CmpTop},
      east = \Lift,
    }
  \end{tikzpicture}
\]

\AxSigmaFiniteJoins*

\begin{proof}
  Finite joins of $\Sigma$-subobjects are preserved by the Yoneda embedding by
  virtue of the localization that we have imposed on $\Psh{\CmpCat}$ in
  \cref{def:dcpo-localization}.
\end{proof}

\NewDocumentCommand\ExtIncl{}{\boldsymbol{e}}

\begin{computation}\label{cmp:lifting-monad-psh-sh}
  The original lifting monad on $\CmpCat$ lifts into $\Psh{\CmpCat}$ by Yoneda extension in the
  following way:
  \[
    \begin{tikzpicture}[diagram]
      \node (D) {$\CmpCat$};
      \node (D') [right = of D] {$\CmpCat$};
      \node (PrD) [right = of D'] {$\Psh{\CmpCat}$};
      \node (PrD') [below = of D] {$\Psh{\CmpCat}$};
      \draw[->] (D) to node [above] {$\Lift$} (D');
      \draw[embedding] (D) to node [left] {$\Yo[\CmpCat]$} (PrD');
      \draw[embedding] (D') to node [above] {$\Yo[\CmpCat]$} (PrD);
      \draw[exists] (PrD') to node [sloped,below] {${\Lift_!}$} (PrD);
    \end{tikzpicture}
  \]

  Moreover ${\Lift_!}$ restricts to the lifting monad on $\Sh{\CmpTop}$
  in the sense that the following diagram commutes up to isomorphism:
  \[
    \DiagramSquare{
      nw = \Sh{\CmpTop},
      sw = \Sh{\CmpTop},
      ne = \Psh{\CmpCat},
      se = \Psh{\CmpCat},
      north/style = embedding,
      south/style = embedding,
      north = \ExtIncl_*,
      south = \ExtIncl_*,
      west = \Lift,
      east = {\Lift_!}
    }
  \]

  This is not difficult to see considering the explicit computation of the
  partial map classifier for a given dominance, using the following facts:
  \begin{enumerate}

    \item The dominance $\Sigma$ is representable, and hence a sheaf because
      our localization was subcanonical.

    \item Sheaves comprise a $\sum$-closed exponential ideal in every slice.

  \end{enumerate}
\end{computation}

The following lemma abstracts the argument of \citet[Theorem~1.7(2) of Fiore and Plotkin]{fiore-plotkin:1996}.

\begin{lemma}\label{lem:lift-preserves-colimits-preserved-by-dir-img}
  The lifting functor $\Mor[\Lift]{\Sh{\CmpTop}}{\Sh{\CmpTop}}$ preserves any
  shape of colimit that is preserved by the direct image functor
  $\EmbMor[\ExtIncl_*]{\Sh{\CmpTop}}{\Psh{\CmpCat}}$.
\end{lemma}

\begin{proof}
  Let $\ICat$ be a category such that colimits of shape $\ICat$ are preserved
  by $\ExtIncl_*$; letting $\Mor[D_\bullet]{\ICat}{\Sh{\CmpTop}}$ be a diagram,
  we compute the colimit of $\Lift{D_\bullet}$.
  \begin{align*}
    &\Hom{\Sh{\CmpTop}}{X}{\Colim{\ICat}{\Lift{D_\bullet}}}
    \\
    &\quad\cong
    \Hom{\Psh{\CmpCat}}{\ExtIncl_*X}{\ExtIncl_*\Colim{\ICat}{\Lift{D_\bullet}}}
    &&\text{$\ExtIncl$ is embedding}
    \\
    &\quad\cong
    \Hom{\Psh{\CmpCat}}{\ExtIncl_*X}{\Colim{\ICat}{\ExtIncl_*\Lift{D_\bullet}}}
    &&\text{by assumption}
    \\
    &\quad\cong
    \Hom{\Psh{\CmpCat}}{\ExtIncl_*X}{\Colim{\ICat}{{\Lift_!}\ExtIncl_*{D_\bullet}}}
    &&\text{\cref{cmp:lifting-monad-psh-sh}}
    \\
    &\quad\cong
    \Hom{\Psh{\CmpCat}}{\ExtIncl_*X}{{\Lift_!}\Colim{\ICat}{\ExtIncl_*{D_\bullet}}}
    &&\text{Yoneda extension is cocontinuous}
    \\
    &\quad\cong
    \Hom{\Psh{\CmpCat}}{\ExtIncl_*X}{\Lift_!\ExtIncl_*\Colim{\ICat}{D_\bullet}}
    &&\text{by assumption}
    \\
    &\quad\cong
    \Hom{\Psh{\CmpCat}}{\ExtIncl_*X}{\ExtIncl_*\Lift\Colim{\ICat}{D_\bullet}}
    &&\text{\cref{cmp:lifting-monad-psh-sh}}
    \\
    &\quad\cong
    \Hom{\Sh{\CmpTop}}{X}{\Lift\Colim{\ICat}{D_\bullet}}
    &&\text{$\ExtIncl$ is embedding}
    \qedhere
  \end{align*}

\end{proof}

\begin{lemma}\label{lem:sheaf-incl-preserves-filtered-colimits}
  Write $\EmbMor[\ExtIncl]{\CmpTop}{\PrTop{\CmpCat}}$ for the embedding of
  topoi that exhibits $\Sh{\CmpTop}$ as the sheaf subcategory spanned by
  presheaves that send the colimits specified in \cref{def:dcpo-localization}
  to limits.  Then the direct image functor
  $\EmbMor[\ExtIncl_*]{\Sh{\CmpTop}}{\Psh{\CmpCat}}$ preserves
  $\omega$-filtered colimits.
\end{lemma}

\begin{proof}
  We consider an $\omega$-filtered category $\ICat$ and a diagram
  $\Mor[D_\bullet]{\ICat}{\Sh{\CmpTop}}$; it is enough to check that the
  colimit of $\Mor[\ExtIncl_*D_\bullet]{\ICat}{\Psh{\CmpCat}}$ is a sheaf.  We
  shall write $Q = \Colim{i\in\ICat}{\ExtIncl_*D_i}$ for this colimit.
  We must verify that finite unions of Scott-open subobjects are sent to limits in
  $\SET$. To that end, fix a finite partition $a \cong \Disj{k}a_k$ of a dcpo $a$
  into Scott-open subdcpos $a_k\subseteq a\in\CmpCat$.
  We must check that $Q\prn{a}$ is isomorphic to
  $\Conj{k}Q\prn{a_k}$.
  We compute:
  \begin{align*}
    Q\prn{a}
    &\cong
    Q\prn{\Disj{k}{a_k}}
    \\
    &\cong
    \Colim{i}{\ExtIncl_*D_i\prn{\Disj{k}{a_k}}}
    &&\text{colimits pointwise in $\Psh{\CmpCat}$}
    \\
    &\cong
    \Colim{i}{\Conj{k}{\ExtIncl_*D_i\prn{a_k}}}
    &&\text{each $\ExtIncl_*D_i$ is a sheaf}
    \\
    &\cong
    \Conj{k}\Colim{i}{\ExtIncl_*D_i\prn{a_k}}
    &&\text{$\omega$-filtered colimits commute with finite limits}
    \\
    &\cong
    \Conj{k}Q\prn{a_k}
    &&\text{colimits pointwise in $\Psh{\CmpCat}$}
    \qedhere
  \end{align*}
\end{proof}

\begin{lemma}\label{lem:lift-preserves-filtered-colimits}
  The functor $\Mor[\Lift]{\Sh{\CmpTop}}{\Sh{\CmpTop}}$ preserves $\omega$-filtered colimits.
\end{lemma}

\begin{proof}
  By \cref{lem:lift-preserves-colimits-preserved-by-dir-img} and
  \cref{lem:sheaf-incl-preserves-filtered-colimits}.
\end{proof}

\begin{remark}
  \citet[Fiore and Plotkin]{fiore-plotkin:1996} show that in their setting the
  lift functor preserves \emph{connected} colimits; this is stronger than
  preserving filtered colimits, and does not hold in our scenario. The reason
  is that they impose only the $0$-extensive topology on $\CmpCat$, hence the
  only product cone that must be considered in their version of
  \cref{lem:sheaf-incl-preserves-filtered-colimits} is the empty one. Because
  we have imposed a different topology, we must consider a broader class of
  limit cones; although it is indeed the case that the empty product is
  preserved by connected colimits, we cannot hope for the same more
  generally~\citep{bjerrum-johnstone-leinster-sawin:2015}.
\end{remark}

\AxOmegaInductive*

\begin{proof}
  By the classic result of \citet[Ad\'amek]{adamek:1974}, it suffices to
  observe that $\Lift$ is $\omega$-cocontinuous by virtue of
  \cref{lem:lift-preserves-filtered-colimits}, hence its initial algebra is the
  colimit of the standard $\omega$-chain above.
\end{proof}

\begin{lemma}\label{lem:representables-are-complete}
  Every representable object in $\Sh{\CmpTop}$ is complete.
\end{lemma}

\begin{proof}
  Fixing a dcpo $d\in\CmpCat$, we must argue that $\TildeYo{d}$ is complete. This
  follows from the limit-colimit coincidence in $\CmpCat$ recalling that $\ShLub
  = \TildeYo{\ShLub}$ is both the initial $\Lift$-algebra and final
  $\Lift$-coalgebra from the perspective of $\CmpCat$. Therefore we are solving
  lifting problems of the following form, recalling that both lifting and
  initial objects are preserved by $\EmbMor[\TildeYo]{\CmpCat}{\Sh{\CmpTop}}$:
  \[
    \begin{tikzpicture}[diagram]
      \node (w) {$I\times \Colim{n}{\TildeYo{\Lift^n\ObjInit{\CmpCat}}}$};
      \node (w*) [below = of w] {$I\times \TildeYo{\Colim{n}{\Lift^n\ObjInit{\CmpCat}}}$};
      \node (d) [right = 3cm of w] {$\TildeYo{d}$};
      \draw[>->] (w) to (w*);
      \draw[->] (w) to (d);
      \draw[exists,->] (w*) to node [sloped,below] {$?$} (d);
    \end{tikzpicture}
  \]

  Applying the fully faithful direct image functor $\EmbMor[\ExtIncl_*]{\Sh{\CmpTop}}{\Psh{\CmpCat}}$, recalling that it commutes with filtered colimits (\cref{lem:sheaf-incl-preserves-filtered-colimits}), we may rewrite our lifting problem to one in $\Psh{\CmpCat}$:
  \[
    \begin{tikzpicture}[diagram]
      \node (w) {$\ExtIncl_*I\times \Colim{n}{\Yo[\CmpCat]{\Lift^n\ObjInit{\CmpCat}}}$};
      \node (w*) [below = of w] {$\ExtIncl_*I\times \Yo[\CmpCat]{\Colim{n}{\Lift^n\ObjInit{\CmpCat}}}$};
      \node (d) [right = 3cm of w] {$\Yo[\CmpCat]{d}$};
      \draw[>->] (w) to (w*);
      \draw[->] (w) to (d);
      \draw[exists,->] (w*) to node [sloped,below] {$?$} (d);
    \end{tikzpicture}
  \]

  Because $\ExtIncl_*I$ is canonically a colimit of representables
  $\Yo[\CmpCat]{e}$ that parameterize an element of $I$, we may rewrite the
  lifting problem once more using the fact that the product functor $\prn{e\times -}$ in both
  $\Psh{\CmpCat}$ and $\CmpCat$ is cocontinuous as both are cartesian
  closed:
  \[
    \begin{tikzpicture}[diagram]
      \node (w) {$\Colim{n}{\Yo[\CmpCat]\prn{e\times \Lift^n\ObjInit{\CmpCat}}}$};
      \node (w*) [below = of w] {$\Yo[\CmpCat]\prn{\Colim{n}{e\times \Lift^n\ObjInit{\CmpCat}}}$};
      \node (d) [right = 3cm of w] {$\Yo[\CmpCat]{d}$};
      \draw[>->] (w) to (w*);
      \draw[->] (w) to (d);
      \draw[exists,->] (w*) to node [sloped,below] {$?$} (d);
    \end{tikzpicture}
  \]

  By the universal property of the colimit in $\Psh{\CmpCat}$, the upper map is
  determined by a cocone
  $\Mor{\Yo[\CmpCat]\prn{e\times \Lift^n\ObjInit{\CmpCat}}}{\brc{\Yo[\CmpCat]{d}}}$ in
  $\Psh{\CmpCat}$, but by the Yoneda lemma this is identical to a cocone
  $\Mor{e\times \Lift^n\ObjInit{\CmpCat}}{\brc{d}}$ in $\CmpCat$. Such a cocone in $\CmpCat$
  uniquely determines the dotted map, via the Yoneda lemma and the universal
  property of the colimit in $\CmpCat$.
\end{proof}

\begin{corollary}\label{lem:representables-are-predomains}
  Every representable object in $\Sh{\CmpTop}$ is a predomain.
\end{corollary}

\begin{proof}
  Representables are closed under lifting, hence the result follows by
  \cref{lem:representables-are-complete}.
\end{proof}

\AxSigmaPredomain*

\begin{proof}
  By \cref{lem:representables-are-predomains}.
\end{proof}

\paragraph{Axioms for information flow and declassification}\label{sec:sdt-info-flow}

We have a geometric morphism $\Mor[\LvlMapCmp]{\CmpTop}{\LvlTop}$; viewing
$\LvlTop$ as the classifying topos for the ``security level theory'' $\LVL$, it
suffices to exhibit a finite meet preserving functor
$\Mor{\LVL}{\Sh{\CmpTop}}$. We will send each $l\in\LVL$ to the subterminal
$\TildeYo{\Yo[\LVL]{l}}$.

\AxLvls*

\begin{proof}
  We consider the Yoneda extension of the following fully faithful and finitely
  continuous composite functor $\EmbMor{\LVL}{\Sh{\CmpTop}}$, noting that the
  image of $\Yo[\LVL]$ lies in $\CmpCat$:
  \[
    \begin{tikzpicture}[diagram]
      \node (Lvl) {$\LVL$};
      \node (D) [right = of Lvl] {$\CmpCat$};
      \node (ShD) [right = of D] {$\Sh{\CmpTop}$};
      \node (PrLvl) [below = of Lvl] {$\Sh{\LvlTop}$};
      \draw[embedding] (Lvl) to node [above] {$\Yo[\LVL]$} (D);
      \draw[embedding] (D) to node [above] {$\TildeYo$} (ShD);
      \draw[embedding] (Lvl) to node [left] {$\Yo[\LVL]$} (PrLvl);
      \draw[->,exists] (PrLvl) to node [sloped,below] {$\LvlMapCmp^*$} (ShD);
    \end{tikzpicture}
  \]

  By Diaconescu's theorem~\citep{diaconescu:1975}, the cocontinuous extension
  indicated above is also finitely continuous, hence the inverse image part of
  a morphism of topoi $\Mor[\LvlMapCmp]{\CmpTop}{\LvlTop}$. It is evident that
  $\EmbMor[\LvlMapCmp^*\Yo[\LVL]]{\LVL}{\Opns{\CmpTop}}$ is valued in
  $\Sigma$-propositions.
\end{proof}

\begin{computation}\label{cmp:lvl-map-cmp-dir-img}
  The direct image functor
  $\Mor[\prn{\LvlMapCmp}_*]{\Sh{\CmpTop}}{\Sh{\LvlTop}}$ is given by
  precomposition with $\Mor|embedding|[\Yo[\LVL]]{\LVL}{\CmpCat}$.
\end{computation}

\begin{proof}
  This is computed by adjointness, fixing $E\in\Sh{\CmpTop}$.
  \begin{align*}
    \prn{\prn{\LvlMapCmp}_*E}l &\cong \Hom{\Sh{\LvlTop}}{\Yo[\LVL]{l}}{\prn{\LvlMapCmp}_*E}\\
    &\cong \Hom{\Sh{\CmpTop}}{\LvlMapCmp^*\Yo[\LVL]{l}}{E}\\
    &\cong \Hom{\Sh{\CmpTop}}{\TildeYo{\Yo[\LVL]{l}}}{E}\\
    &\cong E\prn{\Yo[\LVL]{l}}
    \qedhere
  \end{align*}
\end{proof}

\begin{lemma}\label{lem:lvl-map-cmp-dir-img}
  The direct image functor
  $\Mor[\prn{\LvlMapCmp}_*]{\Sh{\CmpTop}}{\Sh{\LvlTop}}$ extends the the
  forgetful functor $\Mor[\vrt{-}]{\CmpCat}{\Sh{\LvlTop}}$ sending a small dcpo
  to its underlying presheaf of points, in the sense that the following diagram
  commutes up to isomorphism:
  \[
    \begin{tikzpicture}[diagram]
      \node (D) {$\CmpCat$};
      \node (ShD) [below = of D] {$\Sh{\CmpTop}$};
      \node (ShL) [right = of D] {$\Sh{\LvlTop}$};
      \draw[embedding] (D) to node [left] {$\TildeYo$} (ShD);
      \draw[->] (D) to node [above] {$\vrt{-}$} (ShL);
      \draw[->] (ShD) to node [sloped,below] {$\prn{\LvlMapCmp}_*$} (ShL);
    \end{tikzpicture}
  \]
\end{lemma}

\begin{proof}
  We fix a small predomain $A\in\CmpCat$ to compute
  $\prn{\LvlMapCmp}_*\TildeYo{A}$:
  \begin{align*}
    \prn{\prn{\LvlMapCmp}_*\TildeYo{A}}\prn{k} &\cong
    \prn{\TildeYo{A}}\prn{\LvlPol{k}}
    &&\text{\cref{cmp:lvl-map-cmp-dir-img}}
    \\
    &\cong
    \Hom{\CmpCat}{\LvlPol{k}}{A}
    &&\text{Yoneda lemma}
    \\
    &\cong
    \Hom{\Sh{\LvlTop}}{\Yo[\LVL]{k}}{\vrt{A}}
    &&\text{subterminal dcpos are discrete}
    \\
    &\cong
    \vrt{A}\prn{k}
    &&\text{Yoneda lemma}
    \qedhere
  \end{align*}
\end{proof}

\AxFinTransparentPredomain*

\begin{proof}
  That $\CmpTop^*\brk{n}$ is a predomain follows from the fact that a finite
  coproduct of the form $\Coprod{i<n}\ObjTerm{\Sh{\CmpTop}}$ is
  \emph{represented} by the dcpo
  $\Compr{i}{i<n}\in\CmpCat\subseteq\Sh{\LvlTop}$, and all representables are
  predomains. That it is $\LvlPol{l}$-transparent follows from the
  discreteness of the representing dcpo as a constant presheaf in
  $\Sh{\LvlTop}$.
\end{proof}

\subsection{Model of synthetic Tait computability}

\subsubsection{\texorpdfstring{$\LvlTop$}{P} as a classifying topos}\label{sec:l-as-cls-top}

Because $\LVL$ is finitely complete, it can be thought of as a \emph{finite
limit theory} --- the essentially algebraic theory of security levels.

\begin{remark}
  In concrete applications we also expect $\LVL$ to have joins, but we are not
  paying attention to these; while this may appear strange, the joins of a
  security lattice are also ignored by \citet[Abadi \etal]{abadi:1999}, which can be seen
  from the fact that a type that is both $l$-protected and $k$-protected
  nonetheless need not be $l\lor k$-protected in the Dependency Core Calculus.
  Furthermore, it is not the case in \opcit that every type is $\bot$-protected
  when $\LVL$ has a bottom element. We simply reproduce this mismatch between
  security levels and the security typing, though it would be interesting to
  remedy it in future work.
\end{remark}

Any finitely continuous functor $\Mor{\LVL}{\ECat}$ can be thought of as a
\emph{model} of the theory $\LVL$. The collection of all topos
models of $\LVL$ is concentrated in the topos $\LvlTop = \PrTop{\LVL}$: by Diaconescu's
theorem~\citep{diaconescu:1975}, a model $\Mor[L]{\LVL}{\Sh{\XTop}}$ corresponds essentially uniquely
to a morphism of topoi $\Mor[\brk{L}]{\XTop}{\LvlTop}$ whose inverse image
functor is obtained by Yoneda extension:
\begin{equation}\label[diagram]{diag:yoneda-extension}
  \begin{tikzpicture}[diagram,baseline=(sw.base)]
    \node (nw) {$\LVL$};
    \node (sw) [below = of nw] {$\Sh{\LvlTop}$};
    \node (ne) [right = of nw] {$\Sh{\XTop}$};
    \draw[->] (nw) to node [above] {$L$} (ne);
    \draw[embedding] (nw) to node [left] {$\Yo[\LVL]$} (sw);
    \draw[exists,->] (sw) to node [sloped,below] {$\brk{L}^*$} (ne);
  \end{tikzpicture}
\end{equation}

The direct image functor $\Mor[\brk{L}_*]{\Sh{\XTop}}{\Sh{\LvlTop}}$ can be computed by adjointness:
\begin{align}
  \brk{L}_*X &\cong \Hom{\Sh{\LvlTop}}{\Yo[\LVL]{-}}{\brk{L}_*X}&&\text{Yoneda}\\
  &\cong \Hom{\Sh{\XTop}}{\brk{L}^*\Yo[\LVL]{-}}{X}&&\brk{L}^*\dashv\brk{L}_*\\
  &\cong \Hom{\Sh{\XTop}}{L{-}}{X} &&\text{\cref{diag:yoneda-extension}}
  \label{cmp:direct-image}
\end{align}

To see that $\brk{L}^*\dashv\brk{L}_*$ we use the fact that $\brk{L}^*$ is cocontinuous.

\subsubsection{The formal approximation topos}

\begin{notation}[Syntactic topos]
  We will write $\SynTop = \PrTop{\TCat}$ for the presheaf topos defined by the
  identification $\Sh{\SynTop} = \brk{\OpCat{\TCat},\SET}$.
\end{notation}

In \cref{axiom:lvls} we have assumed a morphism of topoi
$\Mor[\LvlMapCmp]{\CmpTop}{\LvlTop}$; viewing $\LvlTop$ as the classifying
topos for the ``theory'' of security levels $\LVL$, this morphism is the
\emph{characteristic map} for the model of $\LVL$ in $\Sh{\CmpTop}$ in which
$\gl{l} = \LvlMapCmp^*\Yo[\LVL]{l}$. Because the syntactic category $\TCat$ of
our programming language also contains security levels, we likewise obtain a
morphism of topoi $\Mor[\LvlMapSyn]{\SynTop}{\LvlTop}$ such that
$\Yo[\TCat]\gl{l} = \LvlMapSyn^*\Yo[\LVL]{l}$.

Writing
$\BaseTop$ for the coproduct topos $\SynTop+\CmpTop$, we may combine these maps
together into a single map $\Mor[\LvlMapBase]{\BaseTop}{\LvlTop}$ along which
glue to obtain the relative Sierpi\'nski cone $\GlTop$:
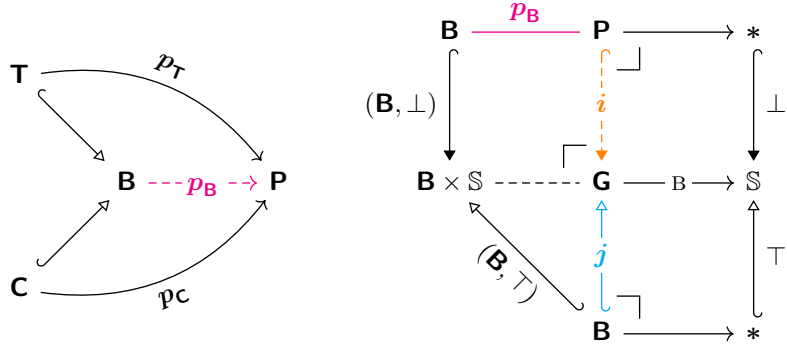
\[
  \begin{tikzpicture}[diagram,baseline=(lvl.base)]
    \node (coprod) {$\BaseTop$};
    \node (lvl) [right = of coprod] {$\LvlTop$};
    \node (T) [above left = of coprod] {$\SynTop$};
    \node (D) [below left = of coprod] {$\CmpTop$};
    \draw[open immersion] (T) to (coprod);
    \draw[open immersion] (D) to (coprod);
    \draw[->,bend left=30] (T) to node [sloped,above] {$\LvlMapSyn$} (lvl);
    \draw[->,bend right=30] (D) to node [sloped,below] {$\LvlMapCmp$} (lvl);
    \draw[->,exists,color=magenta] (coprod) to node [upright desc] {$\LvlMapBase$} (lvl);
  \end{tikzpicture}
  \qquad
  \begin{tikzpicture}[diagram,baseline=(l/sw.base)]
    \SpliceDiagramSquare<l/>{
      nw = \BaseTop,
      ne = \LvlTop,
      sw = \BaseTop\times \SpTop,
      north = \LvlMapBase,
      se = \GlTop,
      west = \prn{\BaseTop,\bot},
      east = \ClImm,
      se/style = pushout,
      west/style = closed immersion,
      east/style = {exists,color={orange},closed immersion},
      east/node/style = upright desc,
      north/style = {color = magenta},
      south/style = exists,
      ne/style = pullback,
    }
    \SpliceDiagramSquare<r/>{
      glue = west, glue target = l/,
      east/style = closed immersion,
      south/node/style = upright desc,
      east = \bot,
      ne = \PtTop,
      se = \SpTop,
      south = \BaseOpn,
    }
    \node (b/se) [below = of r/se] {$\PtTop$};
    \node (b/sw) [sw pullback, below = of l/se] {$\BaseTop$};
    \draw[->] (b/sw) to (b/se);
    \draw[open immersion*,color=cyan] (b/sw) to node [upright desc] {$\OpImm$} (l/se);
    \draw[open immersion*] (b/se) to node [right] {$\top$} (r/se);
    \draw[open immersion*] (b/sw) to node [sloped,below] {$\prn{\BaseTop,\top}$} (l/sw);
  \end{tikzpicture}
\]

\begin{intuition}
  The purpose of constructing $\GlTop$ is as follows: in the category
  $\Sh{\GlTop}$ we may define a Kripke logical relation between the syntax and
  the denotational semantics, with the Kripke variation taken in $\LVL$. Indeed,
  we will see that every object of $\Sh{\GlTop}$ is a kind of ``formal
  approximation relation'' between something syntactical and something
  domain-theoretic.
\end{intuition}

In the construction above, we recover $\BaseTop$ as an \DefEmph{open subtopos} of
$\GlTop$ and $\LvlTop$ as a \DefEmph{closed subtopos}; the inclusions of these
subspaces are given by open and closed immersions $\OpImm,\ClImm$ respectively.
Such an open-closed partition corresponds essentially uniquely to a
characteristic map that we name $\Mor[\BaseOpn]{\GlTop}{\SpTop}$; it is
well-known in topos theory that the open $\BaseOpn$ may be identified with a
subterminal object $\BaseOpn\in\Opns{\GlTop}$ such that $\Sh{\BaseTop}\simeq
\Sl*{\Sh{\GlTop}}{\BaseOpn}$ and under this identification, the inverse image
functor $\Mor[\OpImm^*]{\Sh{\GlTop}}{\Sh{\BaseTop}}$ is the base change functor
$\Mor[\BaseOpn^*]{\Sh{\GlTop}}{\Sl*{\Sh{\GlTop}}{\BaseOpn}}$.

\AxPhases*

\begin{proof}
  The coproduct injections $\Mor|open immersion|{\SynTop,\CmpTop}{\BaseTop}$
  are (cl)open immersions and therefore induce disjoint opens
  $\Mor[\SynOpn,\CmpOpn]{\GlTop}{\SpTop}$ such that $\BaseOpn =
  \SynOpn\lor\CmpOpn$ and $\SynOpn\land\CmpOpn=\bot$.
\end{proof}

Consequently we may likewise identify
$\Sh{\SynTop},\Sh{\CmpTop}$ with
$\Sl*{\Sh{\GlTop}}{\SynOpn},\Sl*{\Sh{\GlTop}}{\CmpOpn}$ respectively.

\begin{remark}[Logical Relations as Types]
  The configuration of the glued topos $\GlTop$ and its open and closed
  subtopoi is in essence a generalization of the topos model of parametricity
  structures by \citet[Sterling and Harper]{sterling-harper:2021}. In \opcit the role of $\LvlTop$
  was played by the Sierpi\'nski topos $\SpTop$, and two copies of the
  syntactic topos were taken as (disjoint) open subspaces whereas here we have
  taken the coproduct of the syntactic and computational topoi. To understand
  the difference between the two constructions, consider that
  \citet[Sterling and Harper]{sterling-harper:2021} were targeting a single phase distinction
  between static and dynamic code, whereas we are concerned with a spectrum of phase
  distinctions corresponding to the security poset $\LVL$; whereas $\SpTop$ is
  the classifying topos of a single phase distinction, $\LvlTop$ is the
  classifying topos for such an array of phase distinctions. Finally, whereas
  \opcit were capturing a binary logical relation between a language and
  itself, we are considering a heterogeneous logical relation between a
  language and its denotational semantics.
\end{remark}

\AxGenericModel*

\begin{proof}
  We may define $\SynAlg$ to be the internal algebra corresponding to the
  (fully faithful) functor $\EmbMor{\TCat}{\Sh{\GlTop}}$ obtained by restricting
  the direct image of the open immersion $\Mor|open
  immersion|{\SynTop}{\GlTop}$ along the Yoneda embedding
  $\EmbMor{\TCat}{\Sh{\SynTop}}$. This does indeed give rise to a model of
  $\TCat$ because both the Yoneda embedding as well as the direct image of an
  open immersion are locally cartesian closed functors.
\end{proof}

\AxSdtInCmpPhase*

\begin{proof}
  This follows immediately by definition, since we have ensured that the open
  subtopos corresponding to the computational phase $\CmpOpn$ is $\CmpTop$
  itself, \ie $\Sl{\GlTop}{\CmpOpn}\simeq \CmpTop$.
\end{proof}

\AxLvlsUnderClosedImmersion*

\begin{proof}
  Working externally, we note that $\CmpAlg.\LvlPol{l} =
  \OpImm_*\TopInr_*\LvlMapCmp^*\Yo[\LVL]{l}$ and $\SynAlg.\LvlPol{l} =
  \OpImm_*\TopInl_*\LvlMapSyn^*\Yo[\LVL]{l}$; to establish our goal, it
  suffices to check that
  $\ClImm^*\OpImm_*\TopInr_*\LvlMapCmp^*\Yo[\LVL]{l} \leq
  \ClImm^*\OpImm_*\TopInr_*\LvlMapSyn^*\Yo[\LVL]{l}$. We compute the left-hand side using \cref{cmp:lvl-map-cmp-dir-img,lem:lvl-map-cmp-dir-img}:
  \[
    \ClImm^*\OpImm_*\TopInr_*\LvlMapCmp^*\Yo[\LVL]{l} \cong
    \prn{\LvlMapBase}_*\TopInr_*\LvlMapCmp^*\Yo[\LVL]{l}
    \cong
    \prn{\LvlMapCmp}_*\LvlMapCmp^*\Yo[\LVL]{l}
    \cong
    \prn{\LvlMapCmp}_*\TildeYo{\Yo[\LVL]{l}}
    \cong \vrt{\Yo[\LVL]{l}}
    \cong \Yo[\LVL]{l}
    \qedhere
  \]
\end{proof}

 \end{appendices}
\fi

\end{document}